\documentclass[a4paper,11pt]{llncs}

\usepackage{amsmath}
\usepackage[utf8]{inputenc}
\usepackage[T1]{fontenc}
\usepackage{amsfonts}
\usepackage{amssymb}
\usepackage{float}
\usepackage{stmaryrd}
\usepackage[margin=2cm]{geometry}
\usepackage{url}
\usepackage{hyperref}
\usepackage{enumitem}
\usepackage{tikz}
  \usetikzlibrary{automata}
  \usetikzlibrary{arrows}
  \usetikzlibrary{shapes}
  \usetikzlibrary{decorations.pathmorphing}
  \usetikzlibrary{fit}
  \usetikzlibrary{trees}
   \usepackage{soul}
  \usetikzlibrary{calc}
  \tikzstyle{every picture}=[
  >=stealth', shorten >=1pt, node distance=1.44cm,auto,bend angle=45,initial text=,
  every state/.style={inner sep=0.75mm, minimum size=1mm},font=\scriptsize,
  triangle/.style = { regular polygon, regular polygon sides=3, inner sep=0}
]

\begin{document}
%%-----------------------------
%%      the top matter
%%-----------------------------
\title{Bottom-Up Derivatives of Tree Expressions}

\author{
  Samira Attou\inst{1}
  \and Ludovic Mignot\inst{2}
  \and Djelloul Ziadi\inst{2}
}
\institute{
  USTHB, Faculty of Mathematics, RECITS Laboratory,\\
  BP 32, El Alia, 16111 Bab Ezzouar, Algiers, Algeria \email{sattou@usthb.dz}
  \and Groupe de Recherche Rouennais en Informatique Fondamentale,\\
  Université de Rouen Normandie,\\ Avenue de l'Université, 76801 Saint-Étienne-du-Rouvray, France.\\
  and Associated Member of RECITS Laboratory, CATI Team, USTHB, Algiers, Algeria.\\
  \email{\{ludovic.mignot,djelloul.ziadi\}@univ-rouen.fr}
}
\date{\today}
\maketitle
\begin{abstract}
  In this paper, we extend the notion of (word) derivatives and partial derivatives due to (respectively) Brzozowski and Antimirov to tree derivatives using already known inductive formulae of quotients.
  We define a new family of extended regular tree expressions (using negation or intersection operators), and we show how to compute a Brzozowski-like inductive tree automaton; the fixed point of this construction, when it exists, is the derivative tree automaton.
  Such a deterministic tree automaton can be used to solve the membership test efficiently: the whole structure is not necessarily computed, and the derivative computations can be performed in parallel.
  We also show how to solve the membership test using our (Bottom-Up) partial derivatives, without computing an automaton.
\end{abstract}

%
% \subjclass{68Q45}
%
%%-----------------------------
%%      your text
%%-----------------------------
\section{Introduction}\label{se:int}
In 1956, Kleene~\cite{Kle56} gave a fundamental theorem in automata theory.
He showed that every regular expression \(E\) can be converted into a finite state machine that recognizes the same language as \(E\), and \emph{vice versa}.
A lot of methods have been proposed to provide the conversion of a given regular expression to a finite word automaton.
One of these approaches which appeared in 1964 was Brzozowski's~\cite{Brzo64} construction; the idea is to use  the notion of derivation to compute a deterministic automaton: the derivative of a regular expression \(E\) w.r.t.\ a word \(w\) is a regular expression that denotes the set of words \(w'\) such that \(ww'\) is denoted by \(E\).
This construction is not necessarily finite: the derivatives of a given regular expression may form an infinite set.
However, considering three equivalence rules (associativity, commutativity and idempotence of the sum), he proved that the set of (so called) similar derivatives is finite.

Antimirov~\cite{Anti96}, in 1996 introduced the partial derivation which is a similar operation to the one defined by Brzozowski; a partial derivative of a regular expression is no longer a regular expression but a set of regular expressions, that leads to the construction of a non-deterministic automaton, with at most~\((n+1)\) states where~\(n\) is the number of letters of the regular expression.
However, this operation is not defined for extended expressions (\emph{i.e.} regular expressions with negation or intersection operators)\footnote{this was achieved by Caron \emph{et al.} using clausal forms instead of sets~\cite{CCM14}}.

Some of these constructions have been extended to tree automata~\cite{ThatcherW68}.
Kuske and Meinecke~\cite{KuskeM11} in 2011, introduced an algorithm to convert a regular tree expression into a non-deterministic tree automaton in a Top-Down interpretation. This construction was inspired by Antimirov's construction.
In 2017, Champarnaud \emph{et al.}~\cite{CMOZ17} have extended the inductive  formulas of quotients to tree languages, following a Bottom-Up interpretation.
These notions of derivatives and quotients have practical and theoretical aspects.
From a practical point of view, this Bottom-Up interpretation can be related to the notion of contexts of trees, that have been studied in  functional programming (\emph{e.g.}  zippers~\cite{Huet97}).
From a theoretical point of view, this study belongs to a large project that aims to study the algorithmic similarities between word automata and tree automata in order to generalize these notions over other algebraic structures~\cite{LLMN19}.

In this paper, we define a new construction of tree automata  based on the notion of derivation of an extended tree expression.
%based on the notion of derivation, from an extended tree expression .
We also show how to extend the notion of partial derivation in a Bottom-Up way and that the previous construction cannot be applied directly with partial derivatives.
Notice that we leave the (technical) study of the finiteness of the set of derivatives and partial derivatives for a future work.

This paper, which is an extended version of~\cite{AMZ19}, is structured as follows: Section~\ref{sec:prelim} defines preliminaries and notations considered throughout this paper.
In Section~\ref{sec:quotients form}, we recall the Bottom-Up quotient formulas for trees and for languages defined in~\cite{CMOZ17}.
We explain in Section~\ref{sec:bool operation} how we can deal with the Boolean operations.
In Section~\ref{sec:tree expression}, we define the derivative formulas for an extended tree expression.
Using the sets of derivatives, in Section~\ref{sec: Auto Construction}, we show how to compute  a deterministic tree automaton from an extended tree expression that recognizes the same language.
We then extend the computation formulae to deal with sets of expressions instead of  a single expression and show their validity in Section~\ref{sec: part der}.
In Section~\ref{sec:appli}, we present a web application allowing the computation of Bottom-Up derivatives, partial derivatives, and both the derivative tree automaton and a classical non-deterministic inductive construction, where the complement is performed \emph{via} determinization.

\section{Preliminaries}\label{sec:prelim}
Let us first introduce some notations and preliminary definitions.
For any pattern matching, we denote by \( \_ \) the wildcard.

In the following of this paper, we consider a \emph{ranked alphabet} \( \Sigma = \bigcup_{k\in\mathbb{N}} \Sigma_k\), \emph{i.e.} a finite graded set of distinct symbols.
A symbol \(f\) in \(\Sigma_k\) is said to be \(k\)\emph{-ary}.

A \emph{tree} \(t\) over \( \Sigma \) is inductively defined by
% \begin{align*}
\(t =\varepsilon_j\) (an empty tree with no symbols) or \(t = f(t_1,\ldots,t_n)\),
% \end{align*}
where \(j\) is a positive integer, \(f\) is a symbol in \(\Sigma_n\), and \(t_1,\ldots,t_n\) are \(n\) trees over \( \Sigma \).
Moreover, we assume that for any integer~\(j\), the symbol \(\varepsilon_j\) appears at most once in a tree.
We denote by \(\mathrm{Ind}_\varepsilon(t)\) the (naturally ordered) set of integers~\(j\) such that \(\varepsilon_j\) appears in the tree~\(t\).
The tree \(t\) is \emph{\(k\)-ary} if \(k\) is the cardinal of \(\mathrm{Ind}_\varepsilon(t)\);
Precisely, \(t\) is \emph{nullary} if \(\mathrm{Ind}_\varepsilon(t)\) is empty.
Given an  integer~\(z\), we denote by \(\mathrm{Inc}_\varepsilon(z,t)\) the substitution of all symbols \(\varepsilon_x\) by \(\varepsilon_{x+z}\) in \(t\).
A tree language (a set of trees) \(L\) is \(k\)\emph{-homogeneous} if it only contains~\(k\)-ary trees \(t\) with the same \( \varepsilon \)-index (ordered) set, denoted by \(\mathrm{Ind}_\varepsilon(L)\) in this case.
The language \(L\) is \emph{homogeneous} if it is \(k\)-homogeneous for some \(k\).
We denote by \(T(\Sigma)\) the set of the trees over \(\Sigma \), and \({T(\Sigma)}_k\) the set of \(k\)-ary trees over \(\Sigma \).
\begin{example}\label{examplePrel}
  Let us consider a ranked alphabets $\Sigma=\Sigma_2\cup\Sigma_1\cup \Sigma_0$ where $f\in\Sigma_2$, $g\in\Sigma_1$ and $a\in\Sigma_0$. Let $t=f(f(t',\varepsilon_3), \varepsilon_1)$  and $t'=g(a)$ be two trees over $\Sigma$. Then, $t$ is $2$-ary because $\vert Ind_{\varepsilon}(t)\vert= \vert\{1,3\}\vert =2$ and $t'$ is $0$-ary  because $\vert Ind_{\varepsilon}(t')\vert= \vert \emptyset \vert=0$.

  Moreover, let us consider the languages $L_1=\{f(a,a), \; a\}$, $L_2=\{f(a,\varepsilon_1), \; \varepsilon_1 \}$ and $L_3=\{f(\varepsilon_1,\varepsilon_2)\}$.  Then, $L_1$ contains only $0$-ary trees than it is $0$-homogeneous, $L_2$ contains only $1$-ary trees then it is $1$-homogeneous and last $L_3$ contains only $2$-ary trees then  it is $2$-homogeneous.
\end{example}
Given a tree~\(t\) over an alphabet \( \Sigma \) with \( \mathrm{Ind}_\varepsilon(t)=\{e_1,\ldots,e_k\} \) and~\(k\) trees \(t_1,\ldots,t_k\) over \(\Sigma \), we denote by \(t\circ(t_1,\ldots,t_k)\) the tree obtained by substituting each \(\varepsilon_{e_i}\) by \(t_i\) in \(t\).
Given a~\(k\)-homogeneous language \(L\) over \(\Sigma \) and~\(k\) languages \((L_1,\ldots,L_k)\) over \(\Sigma \), we denote by~\(\circ \) the operation defined by
\begin{equation}\label{def compo lang}
  L\circ(L_1,\ldots,L_k) = \{t\circ(t_1,\ldots,t_k) \mid (t,t_1,\ldots,t_k) \in L\times L_1 \times \cdots \times L_k\}.
\end{equation}
Let \(L\) be a \(1\)-homogeneous language with \( \mathrm{Ind}_\varepsilon(L)=\{j\} \).
We denote by \(L^n\) the language inductively defined by
% \begin{align*}
\(L^0 =
% \begin{cases}
\{\varepsilon_j\},
L^n = L\circ L^{n-1}\),
% \end{cases}
% \end{align*}
for any integer~\(n > 0\), and we set
\begin{equation*}
  L^\circledast = \bigcup_{n\in\mathbb{N}} L^n.
\end{equation*}
Given a tree \(t\), a symbol \(a\) in \(\Sigma_0\) and a \(0\)-homogeneous language \(L'\), we denote by \(t\cdot_a L'\) the tree language inductively defined by
\begin{gather*}
  \begin{aligned}
    b\cdot_a L'                     & =
    \begin{cases}
      L'    & \text{ if } a = b, \\
      \{b\} & \text{ otherwise,}
    \end{cases}       &
    \qquad \varepsilon_j \cdot_a L' & = \{\varepsilon_j\},
  \end{aligned}\\
  f(t_1,\ldots,t_n) \cdot_a L' = f(t_1\cdot_a L', \ldots, t_n\cdot_a L'),
\end{gather*}
with~\(b\) a tree in \(\Sigma_0\), \(f\) a symbol in \(\Sigma_n\), \(f(L_1,\ldots,L_n) = \{f(t_1,\ldots,t_n) \mid (t_1,\ldots,t_n) \in L_1\times\cdots\times L_n\} \) and~\(n\) trees \(t_1,\ldots,t_n\) over \(\Sigma \).
Moreover, given a homogeneous language \(L\), we set
\begin{equation*}
  L\cdot_a L' = \bigcup_{t\in L} t\cdot_a L'.
\end{equation*}
Finally, let us denote by \(L^{a,n}\) the language inductively defined by
\( L^{a,0} = \{a\} \) and \(L^{a,n} =L^{a,n-1} \cup L\cdot_a L^{n-1}\) for any integer \(n>0\),
and we set
\begin{equation*}
  L^{*_a} = \bigcup_{n\in\mathbb{N}} L^{a,n}.
\end{equation*}
\begin{example} \label{examplee_Lang}
  Let us consider the $1$-homogeneous languages $L=\{f(a,\varepsilon_1)\}$ and $L_1=\{g(\varepsilon_1), \varepsilon_1\}$  over $\Sigma=\Sigma_2\cup\Sigma_1\cup\Sigma_0$ where $f\in\Sigma_2$, $g\in \Sigma_1$ and $a\in\Sigma_0$.

  \begin{align*}
    L\circ (L_1) & =\{f(a,g(\varepsilon_1)), \; f(a,\varepsilon_1) \}, \qquad & L^{a,0}        & = \{ a \},                                                                                        & \\
    L^{a,1}      & =  L^{a,0} \cup \{ f(a,\varepsilon_1) \} \cdot_a L^{a,0}   & \qquad L^{a,2} & =L^{a,1}\cup f(a,\varepsilon_1)\cdot_a L^{a,1}                                                    & \\
                 & =   \{a\} \cup \{ f(a,\varepsilon_1) \} \cdot_a \{a\}      & \qquad         & =\{a, \; f(a,\varepsilon_1) \} \cup \{f(a,\varepsilon_1), \; f(f(a,\varepsilon_1),\varepsilon_1\} & \\
                 & =\{ a, \; f(a,\varepsilon_1)  \},                          & \qquad         & =\{ a, \; f(a,\varepsilon_1), \; f(f(a,\varepsilon_1),\varepsilon_1)) \},                           \\
  \end{align*}
  $$L^{*_a}=\{a,\; f(a,\varepsilon_1), \; f(f(a,\varepsilon_1),\varepsilon_1)), \ldots\}. $$
  \begin{align*}
    L^{0} & = \{\varepsilon_1 \},  \qquad                                      & L^{1} & = L \circ L^0                                         & \\
    L^{2} & = L \circ L^1 \qquad                                               &       & =   \{ f(a,\varepsilon_1) \} \circ \{\varepsilon_1 \} & \\
          & =\{ f(a,\varepsilon_1) \} \circ \{ f(a,\varepsilon_1) \}    \qquad &       & = \{ f(a,\varepsilon_1) \},                           & \\
          & = \{f(a,f(a,\varepsilon_1))\},                                     &       &                                                       &
  \end{align*}
  $$L^\circledast=\{ \varepsilon_1, \; f(a,\varepsilon_1), \; f(a,f(a,\varepsilon_1)), \ldots \}.$$
\end{example}
A \emph{tree automaton} over \(\Sigma \) is a \(4\)-tuple \(\mathrm{A}=(\Sigma, Q, F, \delta)\) where \(Q\) is a set of states, \(F\subseteq Q\) is the set of final states, and \(\delta\subseteq\bigcup_{k\geq 0} (Q^k\times \Sigma_k\times Q)\) is the set of transitions, which can be seen as the function from \(Q^k \times \Sigma_k\) to \(2^Q\) defined by
\begin{equation*}
  (q_1,\ldots,q_k,f,q) \in \delta \Leftrightarrow q \in \delta(q_1,\ldots,q_k,f).
\end{equation*}
It can be linearly extended as the function from \({(2^{Q})}^k \times \Sigma_k\) to \(2^Q\) defined by
\begin{equation}\label{eq:extDeltaEns}
  \delta(Q_1,\ldots,Q_k,f) = \displaystyle \bigcup_{(q_1,\ldots,q_k)\in Q_1\times\cdots Q_k} \delta(q_1,\ldots,q_k,f).
\end{equation}
Finally, we also consider the function \(\Delta \) from \(T(\Sigma)\) to \(2^Q\) defined by
\begin{equation}\label{eq def Delta}
  \Delta(f(t_1,\ldots,t_n)) = \delta(\Delta(t_1),\ldots,\Delta(t_n),f).
\end{equation}
Using these definitions, the language \(L(A)\) recognized by the tree automaton \(A\) is the language \( \{t\in T(\Sigma) \mid \Delta(t)\cap F \neq\emptyset \} \).

A tree automaton \(A=(\Sigma,Q, F,\delta)\) is \emph{deterministic} if for any symbol \(f\) in \(\Sigma_m\), for any \(m\) states \(q_1,\ldots,q_m\) in \(Q\), \(|\delta(q_1,\ldots,q_m,f)|\leq 1\).

\section{Tree Language Quotients}\label{sec:quotients form}

In this section, we recall the inductive definition of the computation of tree quotients defined in~\cite{CMOZ17}.

Let \( (t,t') \) be two trees in \( {T(\Sigma)}_k\times {T(\Sigma)}_{k'} \) such that =\( \mathrm{Ind}_\varepsilon(t)\subseteq \mathrm{Ind}_\varepsilon(t') \).
Let \( R=\mathrm{Ind}_\varepsilon(t) \), \( R'=\mathrm{Ind}_\varepsilon(t') \),
and  \( \{{(x_z)}_{1\leq z\leq k'-k}\}=R'\setminus R \).
The \emph{quotient of} \( t' \) w.r.t. \( t \) is the   \( (k'-k+1) \)-homogeneous tree language \( t^{-1}(t') \) that contains all the trees \( t'' \) satisfying the two following conditions:
\begin{align}
  %\begin{split}
  t'= t''\circ(t,{(\varepsilon_{x_z})}_{  1\leq z\leq k'-k}),
  \quad
  \mathrm{Ind}_\varepsilon(t'')=\{1,{(x_z+1)}_{1\leq z\leq k'-k}\}\label{eq def quot tree}
  %\end{split}
\end{align}
As a direct consequence,

\begin{minipage}{0.45\linewidth}
  \begin{align}
    \varepsilon_j^{-1}(\varepsilon_l) & =
    \begin{cases}
      \varepsilon_1 & \text{ if }j=l,   \\
      \emptyset     & \text{otherwise.}
    \end{cases}
    \label{eq def quot eps}
  \end{align}
\end{minipage}
\hfill
\begin{minipage}{0.45\linewidth}
  \begin{align}
    t^{-1}(t')=\{\varepsilon_1\} & \Leftrightarrow t=t'. \label{eq def quot t par t}
  \end{align}
\end{minipage}
\begin{definition}\label{def quot lang}
  The \emph{Bottom-Up quotient} \(t^{-1}(L)\) of a tree language \( L \) w.r.t.\ a tree~\( t \) is the tree language \(\bigcup_{t'\in L} t^{-1}(t')\).
\end{definition}
\begin{example}\label{Ex-quo-Tre}
  Let us consider the graded alphabet defined by $\Sigma_2=\{f\}$, $\Sigma_1=\{g\}$ and $\Sigma_0=\{a\}$. Let $t=g(a)$ and $t'= f(f(g(a),\varepsilon_1)),g(a))$   be two trees  over $\Sigma=\Sigma_0 \cup \Sigma_1 \cup \Sigma_2$. Then

  $$t^{-1}(t')=f(f(\varepsilon_1,\varepsilon_2),g(a)),\; f(f(g(a),\varepsilon_2),\varepsilon_1) \}. $$
  Notice that for any tree~$t''$ from the set~$t^{-1}(t')$, \; $t'' \circ (g(a),\varepsilon_1)=f(f(g(a),\varepsilon_1)),g(a)).$
\end{example}
As a direct consequence of Equation~\eqref{eq def quot t par t},
the membership of a tree in a tree language can be restated in terms of a quotient:
\begin{proposition}[\cite{CMOZ17}]\label{prop eq membership t eps}
  A tree \(t\) is in a language \(L\) if and only if \(\varepsilon_1\) is in \(t^{-1}(L)\).
\end{proposition}
Let us now make explicit the inductive computation formulae for this quotient operation.
The base cases are the three following ones.
\begin{proposition}[Proposition~$7$ of~\cite{CMOZ17}]\label{Ind-Form-Tree}
  Let \(\Sigma \) be a ranked alphabet,
  \(k\) be an integer,
  and \(\alpha \) be in \(\Sigma_k\):
  \begin{align*}
    \alpha^{-1}(\varepsilon_x)     & =\emptyset, \qquad
    \alpha^{-1}(\alpha(\varepsilon_1,\ldots,\varepsilon_n)) = \{\varepsilon_1\},                                                                 \\
    \alpha^{-1}(f(t_1,\ldots,t_n)) & =\bigcup_{1\leq j\leq n} f(\{t'_1\},\ldots,\{t'_{j-1}\},\alpha^{-1}(\{t_j\}),\{t'_{j+1}\},\ldots,\{t'_n\}),
  \end{align*}
  where \(x\) is an integer in \(\mathbb{N}\), \(f\) is a symbol in \(\Sigma_n\), \(t_1,\ldots,t_n\) are \(n\) trees in \(T_\Sigma \) distinct from \((\varepsilon_1,\ldots,\varepsilon_n)\) and for all integer~\(1\leq z\leq n\), \(t'_z\) is the tree \(\mathrm{Inc}_\varepsilon(1,t_{z})\).
\end{proposition}
By Equation~(\ref{eq def quot tree}) and Definition~\ref{def quot lang}, quotienting by an indexed $\varepsilon$ is reindexing  all the indexed $\varepsilon$ in the language.
\begin{proposition}[Proposition~$9$ of~\cite{CMOZ17}]\label{Ind-Form-Lang}
  Let \(L\) be homogeneous with \(\mathrm{Ind}_\varepsilon(L)=\{j_1,\ldots,j_k\} \) and \(j\) be an integer:

  \begin{equation}
    \resizebox{.9\linewidth}{!}{
      \(
      \varepsilon_j^{-1}(L)=
      \begin{cases}
        L\circ (
        \varepsilon_{j_1+1},\ldots,\varepsilon_{j_{z-1}+1},
        \varepsilon_1,
        \varepsilon_{j_{z+1}+1},\ldots,\varepsilon_{j_{k}+1})
                  & \text{ if } j=j_z\in\mathrm{Ind}_\varepsilon(L), \\
        \emptyset & \text{otherwise.}
      \end{cases}
      \)
    }
    \label{eq quot eps}
  \end{equation}
\end{proposition}
\begin{example}\label{exampl-tree-der}
  Let us consider a tree $t=f(\varepsilon_2,f(a,a))$  with $f\in\Sigma_2$ and $a\in\Sigma_0$. Let us calculate $t^{-1}(t)$. Then
  \begin{align*}
    a^{-1}(t)                       & = \{f(\varepsilon_3,f(\varepsilon_1,a)), \; f(\varepsilon_3,f(a,\varepsilon_1))  \}                                                          \\
    a^{-1}(a^{-1}(t))               & =\{f(\varepsilon_4,f(\varepsilon_2,\varepsilon_1)), \; f(\varepsilon_4,f(\varepsilon_1,\varepsilon_2)) \circ (\varepsilon_1,\varepsilon_2)\} \\
    f(a,a)^{-1}(t)                  & =\{f(\varepsilon_5,\varepsilon_1)\circ (\varepsilon_1,\varepsilon_2) \}                                                                      \\
                                    & = \{f(\varepsilon_2,\varepsilon_1)\}                                                                                                         \\
    f(\varepsilon_2,f(a,a))^{-1}(t) & = \{\varepsilon_1 \}.
  \end{align*}
\end{example}
As a direct consequence of Definition~\ref{def quot lang}, the Bottom-Up quotient for the union of languages can be computed as follows:
\begin{lemma}[Lemma~$13$ of~\cite{CMOZ17}]
  \noindent Let \(t\) be a tree in \(T(\Sigma)\), \(L_1\) and \(L_2\) be two languages over \(\Sigma \).
  Then:
  \begin{equation}
    t^{-1}(L_1\cup L_2)=t^{-1}(L_1)\cup t^{-1}(L_2).
    \label{eq:formuleDerUnion}
  \end{equation}
\end{lemma}
\begin{corollary}[Corollary~$14$ of~\cite{CMOZ17}]
  Let \( t=f(t_1,\ldots,t_k) \) be  an~\( l \)-ary tree such that \( f \) is in \( \Sigma_k \) and \( (t_1,\ldots,t_k) \) is a \( k \)-tuple of trees in \( T(\Sigma) \) different from \( (\varepsilon_1,\ldots,\varepsilon_k) \).
  Let \( L \) be a \( n \)-homogeneous tree language over \( \Sigma \) with \( \mathrm{Ind}_\varepsilon(L)=\{x_1,\ldots,x_n\} \).
  Let \( \{y_1,\ldots,y_{n-l}\}=\mathrm{Ind}_\varepsilon(L)\setminus \mathrm{Ind}_\varepsilon(t) \)
  and \( \forall 1\leq j\leq  k \), \( t'_j=\mathrm{Inc}_\varepsilon(k-j,t_j) \). Then:
  \begin{equation}\label{eq quot lang by tree}
    t^{-1}(L)=(f^{-1}({t'_1}^{-1}(\cdots ({t'_k}^{-1}(L))\cdots))\circ(\varepsilon_1,{(\varepsilon_{y_z+1})}_{1\leq z\leq n-l}).
  \end{equation}
\end{corollary}

The Bottom-Up quotient for the $b$-product of languages can be computed as follows:
\begin{proposition}[Proposition~$17$ of~\cite{CMOZ17}]
  Let \( \Sigma \) be an alphabet.
  Let \( L_1 \) be a \( k \)-homogeneous language, \( L_2 \) be a \( 0 \)-homogeneous language,
  \( \alpha \) be a symbol in \( \Sigma \) and \( b \) be a symbol in \( \Sigma_0 \).
  Then:
  \begin{align*}
    \alpha^{-1}(L_1\cdot_b L_2) & =
    \begin{cases}
      (b^{-1}(L_1)\cdot_b L_2) \circ_1 b^{-1}(L_2)                                       & \text{ if }\alpha=b,                         \\
      \alpha^{-1}(L_1)\cdot_b L_2 \cup (b^{-1}(L_1)\cdot_b L_2) \circ_1 \alpha^{-1}(L_2) & \text{ if }\alpha\in\Sigma_0\setminus \{b\}, \\
      \alpha^{-1}(L_1)\cdot_b L_2                                                        & \text{otherwise,}                            \\
    \end{cases}
  \end{align*}
  where \(\circ_1\) is the partial composition defined by
  \(
  L \circ_1 L' = L\circ (L', {(\varepsilon_l)}_{j_2 \leq l \leq j_k})
  \)
  with \(\mathrm{Ind}_\varepsilon(L)= \{j_1,\ldots,j_k\} \).
\end{proposition}
The Bottom-Up quotient for the composition of languages can be computed as follows:
\begin{proposition}[Proposition~$20$ of~\cite{CMOZ17}]
  Let \( \Sigma \) be an alphabet. Let \( L \) be a \( k \)-homogeneous language with $\mathrm{Ind}_\varepsilon(L)=\{j_1,\ldots,j_k\}$,  $L_1,\ldots,L_k$ be  $k$ tree languages and \( \alpha \) be in $\Sigma_n$. Then:
  \noindent
  %\begin{minipage}{\linewidth}
  \begin{align}
    \begin{split}
      \alpha^{-1}(L\circ(L_1,\ldots,L_k))
      &=
      \displaystyle \bigcup_{1\leq j\leq k} L\circ({(\mathrm{Inc}_\varepsilon(1,L_l))}_{1\leq l\leq j},\alpha^{-1}(L_j),{(\mathrm{Inc}_\varepsilon(1,L_l))}_{j+1\leq l\leq k})\\
      & \qquad \cup
      \begin{cases}
        {\alpha({(\varepsilon_{j_{p_l}})}_{1\leq l\leq n})}^{-1}(L) \circ(\varepsilon_1,{(\mathrm{Inc}_\varepsilon(1,L_l))}_{1\leq l\leq k\mid \forall z, l\neq p_z})) \\
        \quad \text{ if } \forall{} 1\leq{} l\leq{} n, \exists{} 1\leq{} p_l\leq{} k, \varepsilon_l \in{} L_{p_l}                                                      \\
        \emptyset{} \quad \text{ otherwise.}
      \end{cases}
    \end{split}
    \label{eq:derivFormCirc}
  \end{align}
  % \end{minipage}
\end{proposition}
Let us explain the above formula,  in order to explain  the quotient of the composition of the set of  $k$-ary trees  $L$ with $k$-homogeneous languages $L_1,\ldots,L_k$ w.r.t. a symbol $\alpha$ we first  explain  $\alpha^{-1}(t\circ (t_1,\ldots,t_k))$.

The composition of a $k$-ary tree $t$ such that  $Ind_{\varepsilon}(t)= \{x_1,\ldots, x_k\}$, with $k$ trees $t_1,\ldots, t_k$ is the action of grafting  these trees to $t$ at the positions where the symbols $\varepsilon_{x_1},\ldots, \varepsilon_{x_k}$ appear. Thus, the obtained tree $t'$ can be seen as a tree with an upper part containing $t$  and lower parts containing exactly the trees $t_1,\ldots,t_k$. Therefore, if $\alpha$ appears  in a lower tree $t_j$, this tree must be quotiented w.r.t. $\alpha$  and  the other parts are  $\varepsilon$-incremented. Moreover, if some  $n$ trees in $t_1,\ldots, t_k$ are equal to $\varepsilon_1,\ldots,\varepsilon_n$, for example  $t_{p_1},\ldots,t_{p_n}$, and if $t' = \alpha(\varepsilon_{x_{p_1}},\ldots,\varepsilon_{x_{p_n}})$ appears in $t$, then $t'$ must be substituted by $\varepsilon_1$ and the other lower trees $t_j$ with $j\neq p_m$, $m\in \{1,\ldots,n\}$  $\varepsilon$-incremented, since the inverse operation produces $t$.
Therefore, by Definition~\ref{def compo lang} we can extend this operation to the case of languages and we find the above formula.

The Bottom-Up quotient for the composition closure  of a language can be computed as follows:
\begin{proposition}[Proposition~$22$ of~\cite{CMOZ17}]
  Let \( L \) be a \( 1 \)-homogeneous language.
  Let \( \alpha \) be a symbol in \( \Sigma_0\cup \Sigma_1 \).
  Then:
  \begin{align*}
    \alpha^{-1}(L^{\circledast}) & =
    \begin{cases}
      (L^\circledast\circ (\alpha^{-1}(L)))\circ(\varepsilon_1, \mathrm{Inc}_{\varepsilon}(1,L^\circledast)) & \text{ if }\alpha\in\Sigma_0, \\
      (L^\circledast\circ (\alpha^{-1}(L)))                                                                  & \text{ otherwise.}
    \end{cases}
  \end{align*}
\end{proposition}
The Bottom-Up quotient for the iterated composition of  a language   can be computed as follows:
\begin{proposition}[Proposition~$24$ of~\cite{CMOZ17}] Let \( L \) be a \( 0 \)-homogeneous language.
  Let \( \alpha \) and \( b \) be two symbols in \( \Sigma_0 \).
  Then:
  \begin{align*}
    \alpha^{-1}(L^{*_b}) & =
    \begin{cases}
      {(b^{-1}(L))}^\circledast\cdot_b L^{*_b}                           & \text{ if }\alpha=b, \\
      ({(b^{-1}(L))}^\circledast \circ (\alpha^{-1}(L))) \cdot_b L^{*_b} & \text{otherwise.}    \\
    \end{cases}
  \end{align*}
\end{proposition}
\section{Boolean (Homogeneous) Operations}\label{sec:bool operation}
In the following, we consider that Boolean operations (such as union, intersection, complement, \emph{etc.}) are not necessarily defined for all combinations of languages.
Instead, we will consider particular restrictions of these operations, based on the combination of homogeneous languages with the same \(\varepsilon \)-indices.

As an example, given a \(k\)-homogeneous language \(L\), we denote by \(\neg L\) the set
\begin{equation}\label{eqdef homogene comp}
  \{t \in {T(\Sigma)}_k \mid t\notin L, \mathrm{Ind}_{\varepsilon}(t) = \mathrm{Ind}_{\varepsilon}(L) \}.
\end{equation}
By similarly restricting the classical union to pairs of languages with the same \(\varepsilon \)-indices, one can redefine any Boolean operator as a classical combination of union and complementation (\emph{e.g.} symmetrical difference, set difference, \emph{etc.}).
Let us show how to compute the Bottom-Up quotient of a complemented language.
\begin{proposition}\label{prop lang complem}
  Let \(L\) be an homogeneous language over \(\Sigma \) and \(t \in  T(\Sigma)\).
  Then
  \begin{equation*}
    t^{-1}(\neg L) = \neg(t^{-1}(L)).
  \end{equation*}
\end{proposition}
\begin{proof}
  Let \(t''\) be a tree in \(T(\Sigma)\) such that
  \begin{equation*}
    \mathrm{Ind}_\varepsilon(t'')=\{1,{(x_z+1)}_{1\leq z\leq k'-k}\}.
  \end{equation*}
  Then
  \begin{align*}
    t'' \in t^{-1}(\neg L)
     & \Leftrightarrow  t''\circ(t, {(\varepsilon_{x_z})}_{1\leq z\leq k-k'}) \in \neg L \\
     & \Leftrightarrow  t''\circ(t, {(\varepsilon_{x_z})}_{1\leq z\leq k-k'}) \notin  L  \\
     & \Leftrightarrow t'' \notin t^{-1}(L)                                              \\
     & \Leftrightarrow t'' \in \neg(t^{-1}(L)).
  \end{align*}
  %\qed%
\end{proof}

As a direct consequence, following Equation~\eqref{eq:formuleDerUnion}, we get the following result.
\begin{corollary}\label{cor:quot op bool}
  Let \((L_1,\ldots,L_k)\) be \(k\)-homogeneous languages with the same \(\varepsilon \)-indices and let \(\mathrm{op}\) be a Boolean operation.
  Then for any tree \(t\) in \(T(\Sigma)\)
  \begin{equation*}
    t^{-1}(\mathrm{op}(L_1,\ldots,L_k)) = \mathrm{op}(t^{-1}(L_1),\ldots,t^{-1}(L_k)).
  \end{equation*}
\end{corollary}

The restriction to homogeneous languages needs a small modification in terms of computation.
Indeed, let us consider Equation~\eqref{eq:derivFormCirc}.
It is necessary to determine whether~\(\varepsilon_j\) belongs to a given language: how can we decide whether it belongs to~\(\neg \emptyset \)?
There are two alternatives of this barred notation, because the complementation function needs to know the~\( \varepsilon \)-index set of the language it complements.
Either we can specify the restriction by parameterizing the operators, or we specify only the   occurrences of the empty set symbol, leading to the consideration of expressions instead of languages.
This is the approach that we will consider in the following: the languages and the expressions will be subtyped w.r.t. the sets of~\(\varepsilon \)-indices.

\section{Extended Tree Expressions}\label{sec:tree expression}
An \emph{extended tree expression} (tree expression for short) \(E\) over \(\Sigma \) is inductively defined by
\begin{gather}
  \begin{aligned}
    E & = f(E_1, \ldots, E_n),         & E & = \varepsilon_j,     &
    E & = \emptyset_{\mathcal{I}},                                  \\
    E & = \mathrm{op}(E_1,\ldots,E_n), &
    E & = E' \circ (E_1,\ldots, E_n),  & E & = E_1^{\circledast},   \\
    E & = E_1 \cdot_a E_2,             & E & = E_1^{*_a},
  \end{aligned}
  \label{eq:derivFormExp}
\end{gather}
where \(f\) is a symbol in \(\Sigma_n\), \((E', E_1, \ldots, E_n)\) are \((n+1)\) tree expressions over \(\Sigma \), \(j\) is a positive integer, \(\mathcal{I}\) is a set of integers, \(\mathrm{op}\) is an~\(n\)-ary Boolean operator and \(a\) is a symbol in \(\Sigma_0\).
We denote the expression \(\emptyset_\emptyset \) by \(\emptyset \).

The set of \(\varepsilon \)-indices \(\mathrm{Ind}_{\varepsilon}(E)\)  of a tree expression \(E\), that we use to distinguish tree expressions (\emph{e.g.} distinct occurrences of \(\emptyset \)), is inductively defined, following Equation~\eqref{eq:derivFormExp}, by
\begin{gather*}
  \begin{aligned}
    \mathrm{Ind}_{\varepsilon}(\varepsilon_j)           & = \{j\},                                                                   & \qquad
    \mathrm{Ind}_{\varepsilon}(\emptyset_{\mathcal{I}}) & = \mathcal{I},                                                                      \\
    \mathrm{Ind}_{\varepsilon}(E_1 \cdot_a E_2)         & = \mathrm{Ind}_{\varepsilon}(E_1)\cup \mathrm{Ind}_{\varepsilon}(E_2),     &
    \mathrm{Ind}_{\varepsilon}(E_1^{\circledast})       & = \mathrm{Ind}_{\varepsilon}(E_1^{*_a}) = \mathrm{Ind}_{\varepsilon}(E_1).          \\
    \mathrm{Ind}_{\varepsilon}(f(E_1, \ldots, E_n))     & = \mathrm{Ind}_{\varepsilon}(\mathrm{op}(E_1, \ldots, E_n))                         \\ &=  \mathrm{Ind}_{\varepsilon}(E' \circ (E_1,\ldots, E_n))\\
                                                        & = \bigcup_{1\leq k\leq n} \mathrm{Ind}_{\varepsilon}(E_k),
  \end{aligned}\\
\end{gather*}
In the following, we restrict the set of tree expressions that we deal with in order to simplify the different computations.
More formally, we define the notion of valid tree expression, that rejects (for instance) non-homogeneous tree expressions:
A tree expression \(E\) is \emph{valid} if it satisfies the predicate \(V(E)\) inductively defined, following Equation~\eqref{eq:derivFormExp}, by
\begin{align*}
  V(\varepsilon_j)               & = V(\emptyset_{\mathcal{I}}) = \mathrm{True},                                                                                                                       \\
  V(f(E_1, \ldots, E_n))         & = (\bigwedge_{1\leq k\leq n}V(E_k)) \wedge (\bigwedge_{1\leq k < k' \leq n} (\mathrm{Ind}_{\varepsilon}(E_k) \cap \mathrm{Ind}_{\varepsilon}(E_{k'}) = \emptyset)), \\
  V(\mathrm{op}(E_1,\ldots,E_n)) & = (\bigwedge_{1\leq k\leq n}V(E_k)) \wedge (\bigwedge_{1\leq k  < n} (\mathrm{Ind}_{\varepsilon}(E_k) = \mathrm{Ind}_{\varepsilon}(E_{k+1}))),                      \\
  V(E' \circ (E_1,\ldots, E_n))
                                 & = V(E') \wedge (\bigwedge_{1\leq k\leq n}V(E_k)) \wedge (\mathrm{Card}(\mathrm{Ind}_{\varepsilon}(E')) = n)                                                         \\
                                 & \qquad \wedge (\bigwedge_{1\leq k < k' \leq n} (\mathrm{Ind}_{\varepsilon}(E_k) \cap \mathrm{Ind}_{\varepsilon}(E_{k'}) = \emptyset)),                              \\
  V(E_1^{\circledast})           & = V(E_1) \wedge (\mathrm{Card}(\mathrm{Ind}_{\varepsilon}(E_1)) = 1),                                                                                               \\
  V(E_1 \cdot_a E_2)             & = V(E_1) \wedge V(E_2) \wedge (\mathrm{Ind}_{\varepsilon}(E_2) = \emptyset),                                                                                        \\
  V(E_1^{*_a})                   & = V(E_1)\wedge (\mathrm{Ind}_{\varepsilon}(E_1) = \emptyset).
\end{align*}
The language \(L(E)\) \emph{denoted} by a valid tree expression \(E\) with an \(\varepsilon \)-index set \(\mathcal{I}\) is inductively defined by
\begin{gather*}
  \begin{aligned}
    L(f(E_1, \ldots, E_n))         & = f(L(E_1), \ldots, L(E_n)),           & L(\varepsilon_j)     & = \{\varepsilon_j\},        \\
    L(\mathrm{op}(E_1,\ldots,E_n)) & = \mathrm{op}'(L(E_1), \ldots,L(E_n)), &
    L(\emptyset_{\mathcal{I}})     & = \emptyset,                                                                                \\
    L(E' \circ (E_1,\ldots, E_n))  & = L(E') \circ (L(E_1), \ldots L(E_n)), & L(E_1^{\circledast}) & = {(L(E_1))}^{\circledast}, \\
    L(E_1 \cdot_a E_2)             & = L(E_1) \cdot_a L(E_2),               & L(E_1^{*_a})         & = {(L(E_1))}^{*_a},
  \end{aligned}
\end{gather*}
where \(f\) is a symbol in \(\Sigma_n\), \((E', E_1, \ldots, E_n)\) are \((n+1)\) tree expressions over \(\Sigma \), \(j\) is a positive integer, \(\mathrm{op}\) is an~\(n\)-ary Boolean operator, \(\mathrm{op}'\) is  an~\(n\)-ary Boolean operation over homogeneous languages with \(\mathcal{I}\) as \(\varepsilon \)-index
set (\emph{e.g. Equation~\eqref{eqdef homogene comp}})  and \(a\) is a symbol in \(\Sigma_0\).
From these definitions, we can define the derivation formulae for valid tree expressions w.r.t.\ symbols and trees as a syntactical transcription of the quotient formulae.
\begin{definition}\label{def deriv form eps}
  Let \(E\) be a valid tree expression and \(j\) be an \( \varepsilon \)-index of \(E\).
  Then \(d_{\varepsilon_j}(E)\) is obtained by incrementing all the \( \varepsilon \)-indices of \(E\) by \(1\) except \(\varepsilon_j\) which is replaced by \(\varepsilon_1\).
\end{definition}
\begin{definition}\label{def deriv form symb}
  Let \(\alpha \) be a symbol in \(\Sigma_n\) and \(F\) be a valid tree expression over \(\Sigma \) containing \( \{1,\ldots,n\} \) as \(\varepsilon \)-indices.
  The \emph{derivative} of \(F\) w.r.t.\ to \(\alpha \) is the expression inductively defined by
  \begin{align*}
    d_\alpha(\emptyset_{\mathcal{I}})                    & = \emptyset_{\{1\} \cup \{i+1 \mid i > n, i\in\mathcal{I}\}},                                                             \\
    d_\alpha(\varepsilon_1)                              & = \emptyset_{\{1\}},                                                                                                      \\
    d_\alpha(\alpha(\varepsilon_1,\ldots,\varepsilon_n)) & = \varepsilon_1,                                                                                                          \\
    d_\alpha(f(E_1,\ldots,E_m))                          & = \sum_{1\leq j\leq n}  f(\underline{E_1}, \ldots, \underline{E_{j-1}},
    d_\alpha(t_j),
    \underline{E_{j+1}},
    \ldots,
    \underline{E_m}),                                                                                                                                                                \\
                                                         & \qquad + \varepsilon_1 \text{ if }\alpha = f \wedge \forall i\leq m, \varepsilon_i\in L(E_i),                             \\
    d_\alpha(\mathrm{op}(E_1,\ldots,E_k))                & = \mathrm{op}(d_\alpha(E_1), \ldots, d_\alpha(E_k)),                                                                      \\
    d_\alpha(E_1\cdot_b E_2)                             & =
    \begin{cases}
      (d_b(E_1)\cdot_b E_2) \circ_1 d_b(E_2)                                 & \text{ if }\alpha=b,                         \\
      d_\alpha(E_1)\cdot_b E_2 + (d_b(E_1)\cdot_b E_2) \circ_1 d_\alpha(E_2) & \text{ if }\alpha\in\Sigma_0\setminus \{b\}, \\
      d_\alpha(E_1)\cdot_b E_2                                               & \text{otherwise,}                            \\
    \end{cases}                                                                                                                                                       \\
    d_\alpha(E\circ(E_1,\ldots,E_k))                     & =  \sum_{1\leq j\leq k} E\circ  (({\underline{E_l})}_{1\leq l\leq j},d_\alpha(E_j),({\underline{E_l})}_{j+1\leq l\leq k}) \\
                                                         & \quad +
    \begin{cases}
      d_{
      \alpha({(\varepsilon_{j_{p_l}})}_{1\leq l\leq n})
      }
      (E) \circ  (\varepsilon_1,{(\underline{E_l})}_{1\leq l\leq k\mid \forall z, l\neq p_z})        \\
      \quad \text{ if } \forall 1\leq l\leq n, \exists 1\leq p_l\leq k, \varepsilon_l \in L(E_{p_l}) \\
      \emptyset_{\mathrm{Ind}_\varepsilon(E\circ(E_1,\ldots,E_k))\setminus \{1,\ldots,n\}} \quad \text{ otherwise,}
    \end{cases}                                                                                                                                                       \\
    d_\alpha(E^{\circledast})                            & =
    \begin{cases}
      (E^\circledast\circ (d_\alpha(E)))\circ(\varepsilon_1, \mathrm{Inc}_{\varepsilon}(1,E^\circledast)) & \text{ if }\alpha\in\Sigma_0, \\
      (E^\circledast\circ (d_\alpha(E)))                                                                  & \text{ otherwise,}
    \end{cases}                                                                                                                                                       \\
    d_\alpha(E^{*_b})                                    & =
    \begin{cases}
      {(d_b(E))}^\circledast\cdot_b E^{*_b}                        & \text{ if }\alpha=b, \\
      ({(d_b(E))}^\circledast \circ (d_\alpha(E))) \cdot_b E^{*_b} & \text{otherwise,}    \\
    \end{cases}
  \end{align*}
  where
  \(d_{\alpha(\varepsilon_{j_1},\ldots,\varepsilon_{j_n})(E)}
  =
  d_{\alpha}(
  d_{\varepsilon_{j_{1} + n - 1}}(
  \cdots
  d_{\varepsilon_{j_{n-1} + 1}}(
  d_{\varepsilon_{j_n}}(E)
  )
  \cdots
  )
  )
  \),
  where  for all integers~\(i\) the expression \(\underline{E_i}\) equals  \(\mathrm{Inc}_\varepsilon(1,E_i)\)  and where
  \(\circ_1\) is the partial composition defined by
  \(
  E \circ_1 E' = E\circ (E', {(\varepsilon_l)}_{  l \in \{j_2,\ldots, j_k \}})
  \)
  with \(\mathrm{Ind}_\varepsilon(E)= \{j_1,\ldots,j_k\} \).
\end{definition}
\begin{definition}\label{def deriv form tree}
  Let \(t=f(t_1,\ldots,t_n)\) be a tree in \(T(\Sigma)\) and \(E\) a valid tree expression over \(\Sigma \) such that
  \(
  \mathrm{Ind}_\varepsilon(t)\subseteq \mathrm{Ind}_\varepsilon(E).
  \)
  The \emph{derivative} of \(E\) w.r.t. \( t \) is the tree expression defined by
  \begin{equation*}
    d_t(E)=(d_f(d_{t'_1}(\cdots (d_{t'_k}(E))\cdots))\circ(\varepsilon_1,{(\varepsilon_{y_z+1})}_{1\leq z\leq n-l})),
  \end{equation*}
  where  \( \{y_1,\ldots,y_{n-l}\} = \mathrm{Ind}_\varepsilon(E)\setminus \mathrm{Ind}_\varepsilon(t) \) and \(\forall 1\leq j\leq  k \), \( t'_j=\mathrm{Inc}_\varepsilon(k-j,t_j) \).
\end{definition}
Notice that the base cases include the one of the derivation of the empty set.
In this case, the only modification that occurs is the index simulating the \(\varepsilon \)-index set of the denoted language.
This is necessary in order to validate Equation~\eqref{eq def quot tree}.
As an example, consider the expression \(E=\neg \emptyset_\emptyset \).
When deriving \(E\) w.r.t.\ a nullary tree \(a\), one must obtain an expression denoting all the trees that  belong  to \(T(\Sigma)\) in which one~\(a\) was removed, that is the set of all the trees with only \(\varepsilon_1\) as an~\(\varepsilon \)-index.
Applying the previously defined formulae:
\begin{equation*}
  d_a(E) = \neg \emptyset_{\{1\}}.
\end{equation*}
When deriving one more time w.r.t. \(a\), one must obtain an expression denoting all the trees with only \(\varepsilon_1\) and \(\varepsilon_2\) as \( \varepsilon\)-indices (obtained from a tree in \(T(\Sigma)\) by removing  two occurrences of \(a\)).
Applying the previously defined formulae:
\begin{equation*}
  d_a(d_a(E)) = \neg \emptyset_{\{1, 2\}}.
\end{equation*}
Finally, when deriving by a binary symbol \(f\), one must obtain an expression denoting all the trees that belong to \(T(\Sigma)\) in which one occurrence of \(f(a,a)\) was removed, that is the set of all the trees with only \( \varepsilon_1\) as \(\varepsilon\)-index.
Applying the previously defined formulae:
\begin{equation*}
  d_f(d_a(d_a(E))) = \neg \emptyset_{\{1\}}.
\end{equation*}

As a direct consequence of the inductive formulae of Section~\ref{sec:quotients form} and of Corollary~\ref{cor:quot op bool}, we get the following theorem.
\begin{theorem}\label{thm lang deriv}
  The derivative of a valid tree expression \(E\) w.r.t.\ to a tree \(t\) denotes \(t^{-1}(L(E))\).
\end{theorem}
\begin{proof}
  Let us proceed in three steps.
  \begin{enumerate}
    \item\label{step1}
    Following Equation~\eqref{eq quot eps} and Definition~\ref{def deriv form eps}, it holds that
    \begin{equation*}
      L(d_{\varepsilon_j}(E)) = \varepsilon_j^{-1}(L(E)).
    \end{equation*}
    \item\label{step2}
    Notice that the derivation formulae of Definition~\ref{def deriv form symb} are syntactical equivalents of the quotient formulae of Section~\ref{sec:quotients form} and of Corollary~\ref{cor:quot op bool} and therefore it can be proved by induction over the structure of \(E\) that
    \begin{equation*}
      L(d_{\alpha}(E)) = \alpha^{-1}(L(E)).
    \end{equation*}
    This reasoning is valid except for the case with indexed occurrence of the  empty set and  with the case of the substitution product.
    As discussed before, the occurrences of the empty set are ``typed'' w.r.t.\ the \(\varepsilon \)-indices set of the language they denote. Therefore, these indices should be modified using Equation~\eqref{eq def quot tree}.
    In the product case, the derivation of an expression by the tree \(\alpha(\varepsilon_{j_1},\ldots,\varepsilon_{j_n})\) has to be considered.
    However, by considering this particular case in the induction, one can check that
    \begin{align*}
      L(d_{\alpha(\varepsilon_{j_1},\ldots,\varepsilon_{j_n})}(E)) & =
      L(
      d_{\alpha}(
      d_{\varepsilon_{j_{1} + n - 1}}(
      \cdots
      d_{\varepsilon_{j_{n-1} + 1}}(
      d_{\varepsilon_{j_n}}(E)
      )
      \cdots
      )
      ))                                                               \\
                                                                   & =
      \alpha^{-1} (L(
      d_{\varepsilon_{j_{1} + n - 1}}(
      \cdots
      d_{\varepsilon_{j_{n-1} + 1}}(
      d_{\varepsilon_{j_n}}(E)
      )
      \cdots
      )
      )),
    \end{align*}
    this last equality obtained by applying the induction step.
    From item~\ref{step1}, it holds that
    \begin{align*}
      \alpha^{-1}(
      d_{\varepsilon_{j_{1} + n - 1}}(
      \cdots
      d_{\varepsilon_{j_{n-1} + 1}}(
      d_{\varepsilon_{j_n}}(E)
      )
      \cdots
      )
      ) & =
      \alpha^{-1}(
      {\varepsilon_{j_{1} + n - 1}}^{-1}(
      \cdots
      {\varepsilon_{j_{n-1} + 1}}^{-1}(
      L(E)
      )
      \cdots
      )
      )
    \end{align*}
    that equals \({\alpha(\varepsilon_{j_1},\ldots,\varepsilon_{j_n})}^{-1}(L(E))\) from Equation~\eqref{eq quot lang by tree}.
    We can conclude following Equation~\eqref{eq:derivFormCirc}.
    \item
          Finally, according to Equation~\eqref{eq quot lang by tree} and item~\ref{step2}, Definition~\ref{def deriv form tree} implies that
          \begin{equation*}
            L(d_t(E)) = t^{-1}(L(E)).
          \end{equation*}
  \end{enumerate}
  %\qed%
\end{proof}

\begin{example}\label{ex:derivative}
  Let us consider the graded alphabet defined by \( \Sigma_2=\{f\} \),  \( \Sigma_1=\{g\} \) and \( \Sigma_0=\{a,b,c\} \) and let \( E \) be the extended tree expression defined by
  \begin{equation*}
    E = E_1 \cdot_a E_2,
  \end{equation*}
  with \(E_1 = \neg({g(a)}^{*_a})\) and \(E_2 = f(f(a,a),a) \).
  Let us show how to calculate the derivative of \( E \) w.r.t. \(t = f(f(a,a),a)\).
  First, let us compute the derivative of \(E_2\) w.r.t. \(t\):
  \begin{align*}
    d_a(E_2)             & =  f(f(\varepsilon_1,a)+f(a,\varepsilon_1),a)+f(f(a,a),\varepsilon_1),                                    \\
    d_a(d_a(E_2))        & =  f(f(\varepsilon_2,\varepsilon_1)+f(\varepsilon_1,\varepsilon_2),a)
    + f(f(\varepsilon_2,a)+ f(a,\varepsilon_2),\varepsilon_1)                                                                        \\
                         & \qquad + f(f(\varepsilon_1,a)+f(a,\varepsilon_1),\varepsilon_2),                                          \\
    d_a(d_a(d_a(E_2)))   & =  f(f(\varepsilon_3,\varepsilon_2)+f(\varepsilon_2,\varepsilon_3),\varepsilon_1)
    + f(f(\varepsilon_3,\varepsilon_1)+f(\varepsilon_1,\varepsilon_3),\varepsilon_2)                                                 \\
                         & \qquad + f(f(\varepsilon_2,\varepsilon_1)+f(\varepsilon_1,\varepsilon_2),\varepsilon_3),                  \\
    d_{f(a,a)}(d_a(E_2)) & = d_{f(\varepsilon_1, \varepsilon_2)}(d_a(d_a(d_a(E_2)))) \circ (\varepsilon_1, \varepsilon_2)            \\
                         & =
    (\emptyset_{\{1,4\}} + \emptyset_{\{1,4\}} + f(\emptyset_1 + \varepsilon_1 ,\varepsilon_4)) \circ (\varepsilon_1, \varepsilon_2) \\
                         & = f(\varepsilon_1 ,\varepsilon_4) \circ (\varepsilon_1, \varepsilon_2)= f(\varepsilon_1, \varepsilon_2),  \\
    d_{t}(E_2)           & = d_{f(\varepsilon_1, \varepsilon_2)} (d_{f(a,a)}(d_a(E_2))) = \varepsilon_1.
  \end{align*}
  Then, in order to reduce the size of the computed tree expressions, let us set
  \begin{align*}
    E'  & = \neg({g(\varepsilon_1)}^{\circledast})\cdot_a E_2, &
    E'' & = \neg(\emptyset_{\{1,2\}})\cdot_a E_2.
  \end{align*}
  Then:
  \begin{align*}
    d_a(E)               & = E' \circ d_a(E_2),                                                                                            \\
    d_{f(a,a)}(d_{a}(E)) & = E'' \circ (f(\varepsilon_1, a),\varepsilon_4 ) \circ  (\varepsilon_1, \mathrm{Inc}_{\varepsilon}(1,d_a(E_2)))
    + E' \circ d_{f(a,a)}(d_a(E_2)) ,                                                                                                      \\
    d_t(E)               & = E'.
  \end{align*}
\end{example}

\section{Tree Automaton Construction}\label{sec: Auto Construction}
In this section, we explain how we can compute a tree automaton from a valid tree expression \(E\) with \(\mathrm{Ind}_\varepsilon(E) = \emptyset \) from an iterated process using the previously defined derivation.

Given a tree expression \(E\) over an alphabet \(\Sigma \), we first compute the set
\begin{equation*}
  D_0(E) = \{d_a(E)\mid a\in\Sigma_0\}.
\end{equation*}
From this set, we compute the tree automaton \(A_0=(\Sigma,D_0(E),F_0,\delta_0)\) where
\begin{align*}
  F_0      & = \{E'\in D_0(E) \mid \varepsilon_1 \in L(E')\}, &
  \delta_0 & = \{(a, d_a(E)) \mid a \in \Sigma_0\}.
\end{align*}
From this step, one can choose a total function \(\mathrm{tree}_0\) associating any tree expression \(E'\) in \(D_0(E)\) with a tree \(t\) such that
\begin{equation*}
  \mathrm{tree}_0(E') = t \Rightarrow d_t(E) = E',
\end{equation*}
by choosing for any tree expression \(E'\) in \(D_0(E)\) a symbol \(a\in\Sigma_0\) such that \((a,E') \in \delta_0\).
From this induction basis, let us consider the transition set \(\delta_n\) inductively defined by
\begin{align}
  \begin{split}
    \delta_n = \{((E'_1, \ldots, E'_m), f, d_t(E)) \mid
    & t = f(\mathrm{tree}_{n-1}(E'_1),\ldots, \mathrm{tree}_{n-1}(E'_m)),\\
    & f \in \Sigma_m,\\
    & E'_1,\ldots,E'_m \in D_{n-1}(E)\}.
    \label{eq:delta}
  \end{split}
\end{align}
Let us consider the set
\(
D_n(E) = D_{n-1}(E) \cup \pi_3(\delta_n)
\),
where \(\pi_3\) is the classical projection defined by
\(
\pi_3(X) = \{z \mid (\_,\_,z) \in X\}
\).
Obviously, one can once again choose a total function \(\mathrm{tree}_n\) associating any tree expression \(E'\) in \(D_n(E)\) with a tree \(t\) such that
\begin{equation*}
  \mathrm{tree}_n(E') = t \Rightarrow d_t(E) = E',
\end{equation*}
by choosing a transition \(((E'_1, \ldots, E'_m), f, E')\) in \(\delta_n\) for any tree expression \(E'\) in \(D_{n}(E)\setminus (D_{n-1}(E))\)
and defining \(t\) as \(f(\mathrm{tree}_{n-1}(E'_1), \ldots, \mathrm{tree}_{n-1}(E'_k))\).
Finally, by considering the set
\begin{equation}
  F_n = \{E' \in D_n(E) \mid \varepsilon_1 \in L(E')\},
  \label{eq:final states}
\end{equation}
we can define the tree automaton \(A_n=(\Sigma,D_n(E),F_n,\delta_n)\).

Let \(A(E)\) be the fixed point, if it exists, of this process (up to the choice of the \(\mathrm{tree}_*\) functions), called the \emph{Bottom-Up derivative tree automaton} of \(E\).

First, notice that the construction leads to a deterministic tree automaton.
Let us then state that the validity of the construction does not depend on the choice of the \(\mathrm{tree}_n\) functions.
By a direct induction over the structure of \(t\), considering the definition of \(\delta \) in Equation~\eqref{eq:delta}, we get the following proposition.
\begin{proposition}\label{prop Delta t dt}
  Let \(E\) be a valid tree expression over an alphabet \(\Sigma \), \(t\) be a nullary tree in \(T(\Sigma) \) and \(A(E)=(\Sigma,Q,F,\delta)\) be a Bottom-Up derivative tree automaton of \(E\).
  Let \(\Delta(t) = \{E'\} \).
  Then
  % \begin{equation*}
  \(L(E') = t^{-1}(L(E))\).
  % \end{equation*}
\end{proposition}

In order to prove this result, let us first show how derivatives behave w.r.t. permutations.
\begin{lemma}\label{lem permut deriv tuple}
  Let \(t_1,\ldots,t_k\) be \(k\) nullary trees.
  Let \(L\) be a nullary language.
  Let \( \pi \) be a permutation over \( \{1,\ldots,k\} \).
  Then
  \begin{equation*}
    t_1^{-1}(t_2^{-1}(\cdots {t_k}^{-1}(L) \cdots )) = t_{\pi(1)}^{-1}(t_{\pi(2)}^{-1}(\cdots {t_{\pi(k)}}^{-1}(L) \cdots )) \circ (\varepsilon_{\pi(1)}, \ldots, \varepsilon_{\pi(k)}).
  \end{equation*}
\end{lemma}
\begin{proof}
  Let \(t\) be a tree.
  Then, considering Equation~\ref{eq def quot tree},
  \begin{align*}
     & \qquad t \in t_1^{-1}(t_2^{-1}(\cdots {t_k}^{-1}(L) \cdots ))                                                                                                    \\
     & \Leftrightarrow \exists t' \in L, t' = t \circ (t_1,\ldots, t_k)                                                                                                 \\
     & \Leftrightarrow \exists t' \in L, t' = t \circ (\varepsilon_{\pi^{-1}(1)}, \ldots, \varepsilon_{\pi^{-1}(k)}) \circ (t_{\pi(1)},\ldots, t_{\pi(k)})              \\
     & \Leftrightarrow t \circ (\varepsilon_{\pi^{-1}(1)}, \ldots, \varepsilon_{\pi^{-1}(k)}) \in t_{\pi(1)}^{-1}(t_{\pi(2)}^{-1}(\cdots {t_{\pi(k)}}^{-1}(L) \cdots )) \\
     & \Leftrightarrow t \in t_{\pi(1)}^{-1}(t_{\pi(2)}^{-1}(\cdots {t_{\pi(k)}}^{-1}(L) \cdots )) \circ (\varepsilon_{\pi(1)}, \ldots, \varepsilon_{\pi(k)})
  \end{align*}
  %\qed%
\end{proof}

Let us now show that if two trees act similarly \emph{via} the quotient operation,  then their action can be interchanged.
\begin{lemma}
  Let \(t_1,\ldots,t_k\) be \(k\) nullary trees.
  Let \(L\) be a nullary language.
  Let \(1 \leq j \leq k\) be an integer and let \(t'_j\) be a nullary tree such that \({t_j}^{-1}(L) = {t'_j}^{-1}(L) \).
  Then
  \begin{equation*}
    t_1^{-1}(\cdots t_j^{-1}(\cdots {t_k}^{-1}(L) \cdots ) \cdots )\\
    = t_1^{-1}(\cdots {t'_j}^{-1}(\cdots {t_k}^{-1}(L) \cdots ) \cdots ).
  \end{equation*}
\end{lemma}
\begin{proof}
  Let us consider the permutation \(\pi'\) that only permutes \(j\) with \(k\) and acts like the identity for the other integers.
  From Lemma~\ref{lem permut deriv tuple}, one can check the following equivalences:
  \begin{align*}
     & t_1^{-1}(\cdots t_j^{-1}(\cdots {t_k}^{-1}(L) \cdots ) \cdots )                                                                                                                                            \\
     & = {t_{\pi'(1)}}^{-1}(\cdots {t_{\pi'(j)}}^{-1}(\cdots {{t_{\pi'(k)}}}^{-1}(L) \cdots ) \cdots ) \circ (\varepsilon_{\pi'(1)}, \ldots,\varepsilon_{\pi'(j)},\ldots, \varepsilon_{\pi'(k)})                  \\
     & = {t_{1}}^{-1}(\cdots {t_{k}}^{-1}(\cdots {{t_{j}}}^{-1}(L) \cdots ) \cdots ) \circ (\varepsilon_{1}, \ldots,\varepsilon_{j-1},\varepsilon_k,\varepsilon_{j+1},\ldots,\varepsilon_{k-1}, \varepsilon_{j})  \\
     & = {t_{1}}^{-1}(\cdots {t_{k}}^{-1}(\cdots {{t'_{j}}}^{-1}(L) \cdots ) \cdots ) \circ (\varepsilon_{1}, \ldots,\varepsilon_{j-1},\varepsilon_k,\varepsilon_{j+1},\ldots,\varepsilon_{k-1}, \varepsilon_{j}) \\
     & = {t_{\pi'(1)}}^{-1}(\cdots {t_{\pi'(j)}}^{-1}(\cdots {{t'_{\pi'(k)}}}^{-1}(L) \cdots ) \cdots ) \circ (\varepsilon_{\pi'(1)}, \ldots,\varepsilon_{\pi'(j)},\ldots, \varepsilon_{\pi'(k)})                 \\
     & = t_1^{-1}(\cdots {t'_j}^{-1}(\cdots {t_k}^{-1}(L) \cdots ) \cdots )
  \end{align*}
  %\qed%
\end{proof}

The repeated application of this lemma trivially leads to the following corollary.
\begin{corollary}\label{cor toute permut quot tuple}
  Let \(t_1,\ldots,t_k\) be \(k\) nullary trees.
  Let \(L\) be a nullary language.
  Let \(t'_1,\ldots,t'_k\) be \(k\) nullary trees such that for any integer \(1\leq j\leq k\),  \({t_j}^{-1}(L) = {t'_j}^{-1}(L) \).
  Then
  \begin{equation*}
    t_1^{-1}(\cdots {t_k}^{-1}(L) \cdots ) = {t'_1}^{-1}(\cdots {t'_k}^{-1}(L) \cdots )
  \end{equation*}
\end{corollary}

Finally, we can prove our proposition.

\begin{proof}  (of Proposition~\ref{prop Delta t dt})
  It is sufficient, by definition of the Bottom-Up derivative tree automaton, to prove by recurrence over an integer~\(n\) that for any tree \(t\) of height at most \(n + 1\), it holds that for \( \Delta_n(t) = \{E'\} \),
  \begin{equation*}
    L(E') = t^{-1}(L(E)).
  \end{equation*}
  \begin{enumerate}
    \item
          By definition of \(\delta_0\), the equality holds for any tree of height \(1\).
    \item
          Let us consider that the hypothesis holds at the rank \(n - 1\).
          Let \(f(t_1,\ldots,t_k)\) be a tree of height \((n + 1)\).
          Obviously, for an integer \( 1 \leq j \leq k\), the height of the tree \(t_j\) is smaller than \(n\) and therefore, by application of  the induction hypothesis, assuming that \(\Delta_{n-1}(t_j) = \{E_j\} \),
          \begin{equation*}
            L(E_j) = {t_j}^{-1}(L(E)).
          \end{equation*}
          Moreover, by definition, \(\delta_{n}\) contains the transition \( ((E_1, \ldots, E_k), f, d_t(E)) \) where \( t = f(\mathrm{tree}_{n-1}(E_1),\ldots, \mathrm{tree}_{n-1}(E_k)) \).

          Let us set, for any integer \(1 \leq j \leq k\), \(\mathrm{tree}_{n-1}(E_j) = t'_j\), and therefore \( t = f(t'_1,\ldots, t'_k) \).

          It is easy to check, from the construction of the function \(\mathrm{tree}_{n-1}\) and by a trivial recurrence, that for any state \(S\) in \(A_{n-1}\),
          \begin{equation*}
            \Delta_{n-1}(\mathrm{tree}_{n-1}(S)) = \{S\}.
          \end{equation*}

          Hence,  for any integer \( 1 \leq j \leq k\),
          \begin{equation*}
            \Delta_{n-1}(t'_j) = \{E_j\}
          \end{equation*}
          and by  induction hypothesis, since the height of \(t'_j\) is at most \(n\), it holds
          \begin{equation*}
            L(d_{t'_j}(E)) = L(E_j).
          \end{equation*}
          Consequently,
          \begin{equation*}
            {t_j}^{-1}(L(E)) =  {t'_j}^{-1}(L(E)).
          \end{equation*}

          Notice that from Definition~\ref{def deriv form tree},
          \begin{equation*}
            d_t(E) =
            d_f(
            d_{t'_1}(
            \cdots
            d_{t'_k}(E)
            \cdots
            )
            ).
          \end{equation*}
          Hence, from Corollary~\ref{cor toute permut quot tuple} and  Theorem~\ref{thm lang deriv},
          \begin{align*}
            L(d_t(E)) & = L (
            d_f(
            d_{t'_1}(
            \cdots
            d_{t'_k}(E)
            \cdots
            )
            ))                                                                 \\
                      & = f^{-1}({t'_1}^{-1}(\cdots {t'_k}^{-1}(L(E)) \cdots)) \\
                      & = f^{-1}({t_1}^{-1}(\cdots {t_k}^{-1}(L(E)) \cdots))   \\
                      & = L(d_{f(t_1,\ldots,t_k)}(E))
          \end{align*}
  \end{enumerate}
  %\qed%
\end{proof}

As a direct consequence of Proposition~\ref{prop Delta t dt}, Proposition~\ref{prop eq membership t eps} and Equation~\eqref{eq:final states}, the following theorem holds.
\begin{theorem}\label{thm bu deriv tree reco bon lang}
  A Bottom-Up derivative tree automaton of a tree expression \(E\) is deterministic and recognizes \(L(E)\).
\end{theorem}
\begin{proof}
  Let \(A = (\Sigma, \_, F, \delta)\) be the Bottom-Up derivative tree automaton of \(E\).
  By construction, \(A\) is deterministic.
  Let \(t\) be a tree over \(\Sigma \).
  Then:
  \begin{align*}
    t \in L(E) & \Leftrightarrow \varepsilon_1 \in t^{-1}(L(E))                           \\
               & \Leftrightarrow \varepsilon_1 \in L(E') \text{ with } \Delta(t) = \{E'\} \\
               & \Leftrightarrow E' \in F \text{ with } \Delta(t) = \{E'\}                \\
               & \Leftrightarrow t \in L(A).
  \end{align*}
  %\qed%
\end{proof}

Notice that when \(\Sigma_n=\emptyset \) for \(n\geq 2\), the three ACI rules (associativity, commutativity and idempotence of the sum) are sufficient to obtain a finite tree automaton, isomorphic to the classical  Brzozowski automaton (for words).
More generally, one can wonder if these three rules are sufficient in order to obtain a finite set of (similar) derivatives.
This (technical) study is the next step of our study.

\begin{example}\label{ex:cons}
  Let us consider the tree expressions defined in Example~\ref{ex:derivative}, \emph{i.e.}
  \begin{align*}
    E   & = E_1 \cdot_a E_2, & E_1 & = \neg({g(a)}^{*_a}),                                &
    E_2 & = f(f(a,a),a),     & E'  & = \neg({g(\varepsilon_1)}^{\circledast})\cdot_a E_2.
  \end{align*}
  Let us show how to compute a derivative tree automaton of \(E\).
  First, let us compute \( A_0 = (\Sigma, D_0(E), F_0, \delta_0) \):
  by definition,
  \(
  D_0(E) = \{d_a(E), d_b(E), d_c(E)\}.
  \)
  Hence,
  \begin{gather*}
    \begin{aligned}
      d_a(E) & = E' \circ d_a(E_2),                            &
      d_b(E) & = d_c(E) = \neg(\emptyset_{\{1\}}) \cdot_a E_2,
    \end{aligned}\\
    \begin{aligned}
      D_0(E)   & = \{ d_a(E), d_b(E)\},                       &
      F_0      & =\{d_b(E)\},                                 &
      \delta_0 & = \{(a, d_a(E)), (b, d_b(E)), (c, d_b(E))\}.
    \end{aligned}
  \end{gather*}
  The tree automaton \(A_0\) is represented in Figure~\ref{fig a0}, where the state \(\emptyset_1\), a sink-state, and its transitions are omitted.
  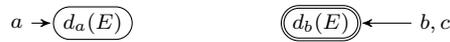
\begin{figure}[H]
    \centering
    \begin{tikzpicture}[node distance=2.5cm,bend angle=30,transform shape,scale=1]
      \node[state, rounded rectangle] (1)  {\( d_a(E) \)} ;
      \node[state, rounded rectangle,double, right of = 1,node distance=3cm] (2) {\( d_b(E) \)};
      %transitions 0
      \draw (1) ++(-1cm,0cm) node {\( a \)}  edge[->] (1);
      \draw (2) ++(1.5cm,0cm) node {\( b,c \)}  edge[->] (2);
    \end{tikzpicture}
    \caption{The Tree Automaton \(A_0\).}%
    \label{fig a0}
  \end{figure}
  We choose a function to define \(\mathrm{tree}_0\):
  \(
  \mathrm{tree}_0(d_a(E)) = a,  \mathrm{tree}_0(d_b(E)) = b.
  \)
  Let us now show how to compute \( A_1=(\Sigma, D_1(E), F_1, \delta_1) \).
  According to Equation~(\ref{eq:delta}), it is sufficient to compute the derivatives of \(E\) w.r.t.\ the trees in the set \( \{f(a,a), f(a,b), f(b,a), f(b,b), g(a), g(b)\} \):
  \begin{gather*}
    \begin{aligned}
      d_{f(a,a)}(E) & = E' \circ f(\varepsilon_1, a), &
      d_{f(b,b)}(E) & = d_{g(b)}(E) = d_{b}(E),
    \end{aligned}\\
    d_{f(a,b)}(E) = d_{f(b,a)}(E) = d_{g(a)}(E) = \emptyset_{\{1\}}.
  \end{gather*}
  There is a new non-final state, \(E' \circ f(\varepsilon_1, a)\), which is associated with \(f(a,a)\) by the function \(\mathrm{tree}_1\), and three new transitions:
  \begin{equation*}
    \delta_1 = \delta_0 \cup \{(d_a(E), d_a(E), f, d_{f(a,a)}(E)), (d_b(E), d_b(E), f, d_{b}(E)), (d_b(E), g , d_b(E))\}.
  \end{equation*}
  The tree automaton \(A_1\) is represented in Figure~\ref{fig a1}.
  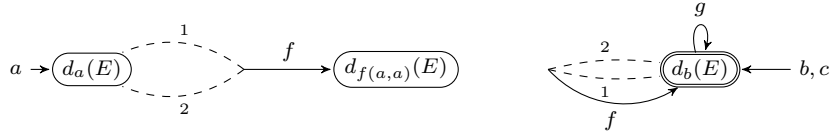
\begin{figure}[H]
    \centering
    \begin{tikzpicture}[node distance=2.5cm,bend angle=30,transform shape,scale=1]
      \node[state, rounded rectangle] (1)  {\( d_a(E) \)} ;
      \node[state, rounded rectangle, right of=1, node distance=4cm] (3) {\( d_{f(a,a)}(E) \)};
      \node[state, rounded rectangle,double, right of=3,node distance=4cm] (2) {\( d_b(E) \)};

      %transitions 0
      \draw (1) ++(-1cm,0cm) node {\( a \)}  edge[->] (1);
      \draw (2) ++(1.5cm,0cm) node {\( b,c \)}  edge[->] (2);
      %transitions 1
      \path[->]
      (2) edge[->, loop above] node {\( g \)} ();
      %transition 2
      %33
      \draw (3) ++(-2 cm,0cm)  edge[->]   node[above,pos=0.5] {\( f \)} (3) edge[dashed, bend left]node[pos=0.5,below]{\tiny{2}} (1) edge[dashed, bend right] node[pos=0.5,above]{\tiny{1}}(1);
      \draw (2) ++(-2cm,0cm) edge[->,bend left=-40]   node[below,pos=0.5] {\( f \)} (2) edge[dashed, bend left=-10]node[pos=0.5,below]{\tiny{1}}(2) edge[dashed, bend left=10] node[pos=0.5,above]{\tiny{2}}(2);
    \end{tikzpicture}
    \caption{The Tree Automaton \(A_1\).}%
    \label{fig a1}
  \end{figure}
  The tree automaton \(A_2\) can be computed through the derivatives w.r.t.\ the trees in the set
  \begin{equation*}
    \{f(f(a,a),f(a,a)), f(f(a,a),a),  f(f(a,a),b), f(a,f(a,a)), f(b,f(a,a)), g(f(a,a))\}:
  \end{equation*}
  \(d_{f(f(a,a),b)}(E) = d_{f(a,f(a,a))}(E) = d_{f(b,f(a,a))}(E) = d_{f(b,f(a,a))}(E) = \emptyset_{\{1\}} \),\\
  \(
  d_{f(f(a,a),a)}(E) = E'\).
  There is a new non-final state, \(E'\), which is associated with \(t = f(f(a,a),a)\) by the function \(\mathrm{tree}_2\), and one new transition:
  \(\delta_2 = \delta_1 \cup \{(d_{f(a,a)}(E), d_a(E), f, d_{t}(E))\} \).
  The tree automaton \(A_2\) is represented in Figure~\ref{fig a2}.
  \begin{figure}[H]
    \centering
    \begin{tikzpicture}[node distance=2.5cm,bend angle=30,transform shape,scale=1]
      \node[state, rounded rectangle] (1)  {\( d_a(E) \)} ;
      \node[state, rounded rectangle, right of=1, node distance=4cm] (3) {\( d_{f(a,a)}(E) \)};
      \node[state, rounded rectangle, right of=3,node distance = 3cm] (4) {\( d_{f(f(a,a),a)}(E) \)};
      \node[state, rounded rectangle,double, below right of=4,node distance=2cm] (2) {\( d_b(E) \)};

      %transitions 0
      \draw (1) ++(-1cm,0cm) node {\( a \)}  edge[->] (1);
      \draw (2) ++(1.5cm,0cm) node {\( b,c \)}  edge[->] (2);
      %transitions 1
      \path[->]
      (2) edge[->, loop above] node {\( g \)} ();
      %transition 2
      %33
      \draw (3) ++(-2 cm,0cm)  edge[->]   node[above,pos=0.5] {\( f \)} (3) edge[dashed, bend left]node[pos=0.5,below]{\tiny{2}} (1) edge[dashed, bend right] node[pos=0.5,above]{\tiny{1}}(1);
      \draw (2) ++(-2cm,0cm) edge[->,bend left=-40]   node[below,pos=0.5] {\( f \)} (2) edge[dashed, bend left=-10]node[pos=0.5,below]{\tiny{1}}(2) edge[dashed, bend left=10] node[pos=0.5,above]{\tiny{2}}(2);
      \draw (4) ++(-1 cm, -1cm)  edge[->,bend right]   node[left,pos=0.5] {\( f \)} (4) edge[dashed, bend right=20]node[pos=0.5,above]{\tiny{2}} (3) edge[dashed, bend left=40] node[pos=0.5,below]{\tiny{1}}(1);
    \end{tikzpicture}
    \caption{The Tree Automaton \(A_2\).}%
    \label{fig a2}
  \end{figure}
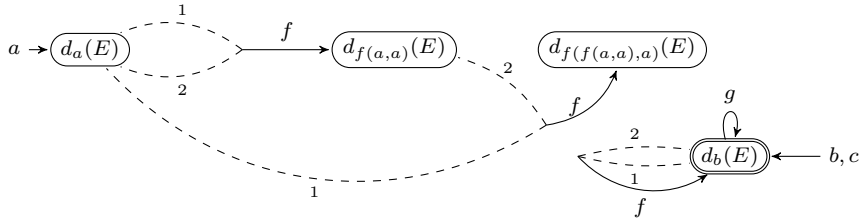
  The tree automaton \(A_3\) is obtained by computing the derivatives w.r.t.\ the trees in the set \( \{f(t, t), f(t, a), f(t, b),\ldots \} \).
  There are only four computations that do not return a derivative equal to \( \emptyset_{\{1\}} \):
  \(d_{f(t,t)}(E) = d_{f(t,b)}(E) = d_{f(b, t)}(E) = d_b(E)\) and
  \(d_{g(t)}(E) = d_t(E)\).
  There are four new transitions:
  \begin{align*}
    \delta_3 & = \delta_2 \cup \{(d_{t}(E), d_t(E), f, d_b(E)), (d_{t}(E), d_b(E), f, d_b(E)), (d_{b}(E), d_t(E), f, d_b(E)), \\
             & \quad (d_{t}(E), g, d_t(E))\},
  \end{align*}
  but these computations does not introduce new states.
  Therefore the computation halts and the derivative tree automaton of \( E \) is \(A_3\), represented in Figure~\ref{fig1}.
\end{example}
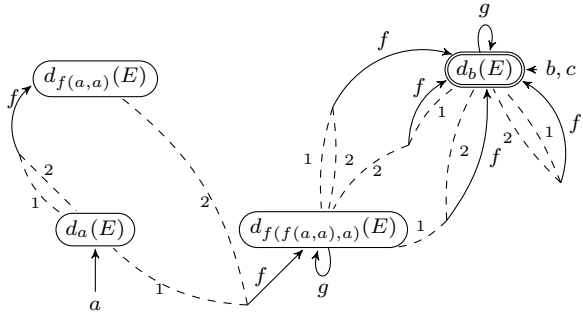
\begin{figure}[H]
  \centering
  \begin{tikzpicture}[node distance=2.5cm,bend angle=30,transform shape,scale=1]
    \node[state, rounded rectangle] (1)  {\( d_a(E) \)} ;
    \node[state, rounded rectangle, above of=1,node distance=2cm] (3) {\( d_{f(a,a)}(E) \)};
    \node[state, rounded rectangle, right of=1,node distance = 3cm] (4) {\( d_{f(f(a,a),a)}(E) \)};
    \node[state, rounded rectangle,double, above right of=4,node distance=3cm] (2) {\( d_b(E) \)};

    %transitions 0
    \draw (1) ++(0cm,-1cm) node {\( a \)}  edge[->] (1);
    \draw (2) ++(1cm,0cm) node {\( b,c \)}  edge[->] (2);
    %transitions 1
    \path[->]
    (2) edge[->, loop above] node {\( g \)} ()
    (4) edge[->, loop below] node {\( g \)} ();
    %transition 2
    %33
    \draw (3) ++(-1 cm,-1cm)  edge[->,bend left=40]   node[above,pos=0.5] {\( f \)} (3) edge[dashed, bend left=-20]node[pos=0.5,below]{\tiny{1}} (1) edge[dashed, bend left=0] node[pos=0.5,above]{\tiny{2}}(1);
    %22
    \draw (4) ++(-1 cm, -1cm)  edge[->,bend left=0]   node[left,pos=0.5] {\( f \)} (4) edge[dashed, bend right=20]node[pos=0.5,below]{\tiny{2}} (3) edge[dashed, bend left=20] node[pos=0.5,left]{\tiny{1}}(1);
    %11
    \draw (2) ++(-2 cm,-0.5cm)  edge[->,bend left=40]   node[above,pos=0.5] {\( f \)} (2) edge[dashed, bend left=10] node[pos=0.5,right]{\tiny{2}}(4) edge[dashed, bend left=-10] node[pos=0.5,left]{\tiny{1}}(4);
    \draw (2) ++(-1cm,-1cm)  edge[->,bend left=30]   node[above,pos=0.5] {\( f \)} (2) edge[dashed, bend left=-20] node[pos=0.5,right]{\tiny{2}}(4) edge[dashed, bend left=15] node[pos=0.5,right]{\tiny{1}}(2);
    \draw (2) ++(-0.5cm,-2cm) edge[->,bend left=-20]   node[right,pos=0.5] {\( f \)} (2) edge[dashed, bend left=20] node[pos=0.5,above]{\tiny{1}}(4) edge[dashed, bend left=15] node[pos=0.5,right]{\tiny{2}}(2);
    \draw (2) ++(1cm,-1.5cm) edge[->,bend left=-40]   node[right,pos=0.5] {\( f \)} (2) edge[dashed, bend left=-10]node[pos=0.5,right]{\tiny{1}}(2) edge[dashed, bend left=10] node[pos=0.5,left]{\tiny{2}}(2);
  \end{tikzpicture}
  \caption{The Bottom-Up Derivative Tree Automaton of \(E\).}%
  \label{fig1}
\end{figure}

\section{Partial Derivatives and Derived Terms}\label{sec: part der}
The partial derivation, due to Antimirov~\cite{Anti96}, is an operation based on the same operation as the derivation.
Indeed, they both implement the computation of the quotient at the expression (and then syntactical) level.
The main difference is that instead of computing an expression from an expression, the partial derivation produces an expression set from an expression.
Computing a sum from this set produces an expression equivalent to the one obtained \emph{via} derivation.
It was already extended to tree expressions by Kuske and Meinecke~\cite{KuskeM11} in order to compute a Top-Down automaton.
Therefore, let us then show how to extend it in a Bottom-Up way.

As in the case of (Bottom-Up) derivation,  let us first explicit the derivation w.r.t. an empty tree. Let \(E\) be a valid tree expression and \(j\) be an \( \varepsilon \)-index of \(E\).
Then \(\partial_{\varepsilon_j}(E)\) is obtained, as in the case of derivation, by incrementing all the \( \varepsilon \)-indices of \(E\) by \(1\) except \(\varepsilon_j\) which is replaced by \(\varepsilon_1\)  \emph{i.e.}
\begin{equation*}\label{def deriv partial eps}
  \partial_{\varepsilon_j}(E)=\{d_{\varepsilon_j}(E) \}.
\end{equation*}
Its application to an expression set can be easily defined as follows:
\begin{equation*}
  \partial_{\varepsilon_j}(\mathcal{E})=\bigcup_{E\in \mathcal{E}} \partial_{\varepsilon_j}(E).
\end{equation*}
The computation of the derived terms, that are the expressions in the set obtained by the partial derivation, can be defined w.r.t.\ the symbols as follows.
Notice that we only deal with the sum as Boolean operator.
Classically, the partial derivation is defined for (so-called) simple expressions; however, it can be easily extended to all Boolean operators~\cite{CCM11}.
\begin{definition}\label{ partial deriv partial symb}
  Let \(\alpha \) be a symbol in \(\Sigma_n\) and \(F\) be a valid tree expression over \(\Sigma \) containing \( \{1,\ldots,n\} \) as \(\varepsilon \)-indices.
  The \emph{partial derivative} of \(F\) w.r.t.\ to \(\alpha \) is the expression inductively defined by
  \begin{align*}
    \partial_\alpha(\emptyset_{\mathcal{I}})                    & =\emptyset,                                                                                                          \\
    \partial_\alpha(\varepsilon_1)                              & = \emptyset,                                                                                                         \\
    \partial_\alpha(\alpha(\varepsilon_1,\ldots,\varepsilon_n)) & = \{\varepsilon_1\},                                                                                                 \\
    \partial_\alpha(f(E_1,\ldots,E_m))                          & = \bigcup_{1\leq j\leq m} \{f(\underline{E_1}, \ldots,\underline{E_{j-1}}, E'_{j},\underline{E_{j+1}},\ldots,
    \underline{E_m})\mid E'_j\in\partial_{\alpha}(E_j)\},                                                                                                                              \\
                                                                & \qquad\cup \{\varepsilon_1 \mid \alpha = f \wedge \forall i\leq m, \varepsilon_i \in L(E_i)\}                        \\
    \partial_\alpha(E_1+E_2)                                    & = \partial_\alpha(E_1)\cup \partial_\alpha(E_2),                                                                     \\
    \partial_\alpha(E_1\cdot_b E_2)                             & =
    \begin{cases}
      (\partial_b(E_1)\cdot_b E_2) \circ_1 \partial_b(E_2)                                          & \text{ if }\alpha=b,                         \\
      \partial_\alpha(E_1)\cdot_b E_2 \cup(\partial_b(E_1)\cdot_b E_2) \circ_1 \partial_\alpha(E_2) & \text{ if }\alpha\in\Sigma_0\setminus \{b\}, \\
      \partial_\alpha(E_1)\cdot_b E_2                                                               & \text{otherwise,}                            \\
    \end{cases}                                                                                                                                                         \\
    \partial_\alpha(E\circ(E_1,\ldots,E_k))                     & =  \bigcup_{1\leq j\leq k}  \{E\circ({(\underline{E_l})}_{1\leq l\leq j},E'_j,{(\underline{E_l})}_{j+1\leq l\leq k}) \\
                                                                & \quad\quad \mid E'_j\in\partial_\alpha(E_j)\}                                                                        \\
                                                                & \quad \cup
    \begin{cases}
      \partial_{
        \alpha({(\varepsilon_{j_{p_l}})}_{1\leq l\leq n})
      }
      (E) \circ(\varepsilon_1,{((\underline{E_l})}_{1\leq l\leq k\mid \forall z, l\neq p_z}),          \\
      \qquad \text{ if } \forall 1\leq l\leq n, \exists 1\leq p_l\leq k, \varepsilon_l \in L(E_{p_l}), \\
      \emptyset \text{ otherwise,}
    \end{cases}                                                                                                                                                         \\
    \partial_\alpha(E^{\circledast})                            & =
    \begin{cases}
      (E^\circledast\circ (\partial_\alpha(E)))\circ(\varepsilon_1, \mathrm{Inc}_{\varepsilon}(1,E^\circledast)) & \text{ if }\alpha\in\Sigma_0, \\
      (E^\circledast\circ (\partial_\alpha(E)))                                                                  & \text{ otherwise,}
    \end{cases}                                                                                                                                                         \\
    \partial_\alpha(E^{*_b})                                    & =
    \begin{cases}
      {(\partial_b(E))}^\circledast\cdot_b E^{*_b}                                 & \text{ if }\alpha=b, \\
      ({(\partial_b(E))}^\circledast \circ_1 (\partial_\alpha(E))) \cdot_b E^{*_b} & \text{otherwise,}    \\
    \end{cases}
  \end{align*}

  where
  \begin{itemize}
    \item \(\partial_{\alpha(\varepsilon_{j_1},\ldots,\varepsilon_{j_n})(E)}
          =
          \partial_{\alpha}(
          \partial_{\varepsilon_{j_{1} + n - 1}}(
          \cdots
          \partial_{\varepsilon_{j_{n-1} + 1}}(
          \partial_{\varepsilon_{j_n}}(E)
          )
          \cdots
          )
          )
          \),

    \item for all  integers~\(i\) the expression \(\underline{E_i}\) equals  \(\mathrm{Inc}_\varepsilon(1,E_i)\),

    \item
          \(
          \mathcal{E}\circ_1 \mathcal{E'} = \{E\circ_1 E'\mid (E,E')\in\mathcal{E}\times\mathcal{E'}\}
          \),

    \item \(\mathcal{E}\circ (E_1,\ldots,E_m)=\{E\circ(E_1,\ldots,E_m)\mid E\in\mathcal{E} \} \),

    \item \( E\circ_1 \mathcal{E} = \{E\circ_1 E' \mid E'\in\mathcal{E} \} \),

    \item \( \mathcal{E}\cdot_b E=\{E'\cdot_b E\mid E'\in\mathcal{E} \} \),

    \item \(\mathcal{E}^{\circledast}= \{{(\sum_{E\in\mathcal{E}}  E)}^{\circledast}\} \).
  \end{itemize}

  The extension of the partial derivation to the operation of negation can be achieved in a similar way to the composition closure: \( \partial(\neg E)=\{\neg(\sum_{E'\in \partial(E)} E')\}
  \).

  Another method is to define another partial derivation returning sets of sets of expressions, interpreted as clausal disjunctive forms (an union of intersection of languages)\cite{CCM11,CCM14}.
\end{definition}
To well define  the operation \(\circledast \), an order over the expressions  has to be considered. This  order  is important in order to distinguish  between semantic and syntax. As an example, let us consider the set $\{a,b\}^{\circledast}$; considering $a<b$ or $b<a$ leads to two distinct  (but equivalent) expressions, $(a+b)^{\circledast}$ and $(b+a)^{\circledast}$.
Moreover, we set \(\emptyset^{\circledast} = \{\varepsilon_1\} \).
Furthermore, the partial derivative of an expression set \(\mathcal{E}\) is defined by
\begin{equation*}
  \partial_{\alpha}(\mathcal{E})=\bigcup_{E\in\mathcal{E}}\partial_{\alpha}(E).
\end{equation*}
Finally, the partial derivative w.r.t.\ a tree can be defined as follows.
\begin{definition}\label{def deriv partial tree}
  Let \(t=f(t_1,\ldots,t_n)\) be a tree in \(T(\Sigma)\) and \(E\) a valid tree expression over \(\Sigma \) such that
  \(
  \mathrm{Ind}_\varepsilon(t)\subseteq \mathrm{Ind}_\varepsilon(E).
  \)
  The \emph{partial derivative} of \(E\) w.r.t. \( t \) is the set of tree expressions defined by
  \begin{equation*}
    \partial_t(E)=(\partial_f(\partial_{t'_1}(\cdots (\partial_{t'_k}(E))\cdots))\circ(\varepsilon_1,{(\varepsilon_{y_z+1})}_{1\leq z\leq n-l})),
  \end{equation*}
  where  \( \{y_1,\ldots,y_{n-l}\} = \mathrm{Ind}_\varepsilon(E)\setminus \mathrm{Ind}_\varepsilon(t) \) and \(\forall 1\leq j\leq  k \), \( t'_j=\mathrm{Inc}_\varepsilon(k-j,t_j) \).
\end{definition}
As a direct consequence of the inductive formulae of Section~\ref{sec:quotients form}, we get the following theorem.
\begin{theorem}\label{thm lang  part-deriv}
  Let \(t\) be a tree in \(T(\Sigma)\) and \(E\) a valid tree expression over \(\Sigma \).
  Then:
  \begin{equation*}
    \bigcup_{E'\in \partial_t(E)} L(E')=t^{-1}(L(E)).
  \end{equation*}
\end{theorem}
\begin{proof}
  The proof is similar to the one of  Theorem~\ref{thm lang deriv},
  by first considering that union distributes over composition (as a direct consequence of Equation~\eqref{def compo lang})
  \begin{multline*}
    L\circ(L_1,\ldots,L_{i-1},L_i \cup L_{i'},L_{i+1},\ldots, L_n)\\
    =
    L\circ(L_1,\ldots,L_{i-1},L_i,L_{i+1},\ldots, L_n)
    \cup
    L\circ(L_1,\ldots,L_{i-1},L_{i'},L_{i+1},\ldots, L_n)
  \end{multline*}
  %\qed%
\end{proof}

%\begin{example}
\begin{example}
  Let us consider the graded alphabet defined by \( \Sigma_2=\{f\} \),  \( \Sigma_1=\{g\} \) and \( \Sigma_0=\{a,b\} \) and let \( E \) be the tree expression defined by

  \begin{equation*}
    E= E_1+E_2,
  \end{equation*}

  with \(E_1 = f(a,a+b) \) and \(E_2 = g(a)^{*_a}\cdot_af(b,a) \).

  Let us show how to calculate the derivative of \( E \) w.r.t. \(t = g(f(b,a))\).
  \begin{align*}
    \partial_a(E)             & = \partial_a(E_1)\cup \partial_a(E_2)                                                                                 \\
                              & =\{\,f(a,\varepsilon_1),\ f(\varepsilon_1,a+b)\,\}\cup  \{\,g(\varepsilon_1)^\circledast\circ(f(b,\varepsilon_1))\,\} \\
                              & =\{f(a,\varepsilon_1),\ f(\varepsilon_1,a+b),\ g(\varepsilon_1)^\circledast\circ(f(b,\varepsilon_1)) \}               \\
    \partial_b(\partial_a(E)) & = \bigcup_{E'\in \partial_a(E)}\partial_b(E')                                                                         \\
                              & =\{\, f(\varepsilon_2,\varepsilon_1),\ g(\varepsilon_1)^{\circledast}\circ(f(\varepsilon_1,\varepsilon_2)) \, \}      \\
    \partial_{f(b,a)}(E)      & = \partial_{f(\varepsilon_1,\varepsilon_2)}(\partial_b(\partial_a(E))                                                 \\
                              & =\{\,g(\varepsilon_1)^\circledast\, \}                                                                                \\
    \partial_{g(f(b,a))}(E)   & =\partial_{g(\varepsilon_1)}(\partial_{f(b,a)}(E))                                                                    \\
                              & =\{ \,g(\varepsilon_1)^\circledast \,\}.
  \end{align*}
  %\end{example}
\end{example}

\begin{example}
  \label{ex: partial derivative}
  Let us consider the graded alphabet defined by \( \Sigma_2=\{f\} \),  \( \Sigma_1=\{g\} \) and \( \Sigma_0=\{a\} \) and let \( E \) be the tree expression defined by
  \begin{equation*}
    E =  ({g(a)}^{*_a}) \cdot_a f(f(a,a),a).
  \end{equation*}
  Let us show how to calculate the derivative of \( E \) w.r.t. \(t = f(f(a,a),a)\).
  \begin{align*}
    \partial_a(E)                          & =\{\, g(\varepsilon_1)^\circledast \circ (f(f(a,\varepsilon_1),a)),\ g(\varepsilon_1)^\circledast \circ (f(f(a,a),\varepsilon_1),\                                                     \\
                                           & \qquad g(\varepsilon_1)^\circledast \circ (f(f(\varepsilon_1,a),a))\, \}                                                                                                               \\
    \partial_a(\partial_a(E))              & = \bigcup_{E'\in\partial_a(E)}\partial_a(E')                                                                                                                                           \\
                                           & =\{\,g(\varepsilon_1)^\circledast \circ (f(f(\varepsilon_1,\varepsilon_2),a)),\ g(\varepsilon_1)^\circledast \circ (f(f(a,\varepsilon_2),\varepsilon_1)),                              \\
                                           & \qquad g(\varepsilon_1)^\circledast \circ (f(f(\varepsilon_1,a),\varepsilon_2)),\ g(\varepsilon_1)^\circledast \circ (f(f(a,\varepsilon_1),\varepsilon_2)),                            \\
                                           & \qquad g(\varepsilon_1)^\circledast \circ (f(f(\varepsilon_2,\varepsilon_1),a),\ g(\varepsilon_1)^\circledast \circ (f(f(\varepsilon_2,a),\varepsilon_1)\,\}                           \\
    \partial_a(\partial_a(\partial_a(E)))= & \bigcup_{E''\in\partial_a(\partial_a(E))} \partial_a(E'')                                                                                                                              \\
                                           & =\{\,g(\varepsilon_1)^\circledast \circ (f(f(\varepsilon_2,\varepsilon_3),\varepsilon_1)),\, g(\varepsilon_1)^\circledast \circ (f(f(\varepsilon_1,\varepsilon_3),\varepsilon_2)),     \\
                                           & \qquad g(\varepsilon_1)^\circledast \circ (f(f(\varepsilon_2,\varepsilon_1),\varepsilon_3)),\  g(\varepsilon_1)^\circledast \circ (f(f(\varepsilon_1,\varepsilon_2),\varepsilon_3)),   \\
                                           & \qquad g(\varepsilon_1)^\circledast \circ (f(f(\varepsilon_3,\varepsilon_2),\varepsilon_1)),\ g(\varepsilon_1)^\circledast \circ (f(f(\varepsilon_3,\varepsilon_1),\varepsilon_2))\,\} \\
    \partial_{f(a,a)}(\partial_a(E))       & =\partial_{f(\varepsilon_1,\varepsilon_2)}(\partial_a(\partial_a(\partial_a(E))))\circ (\varepsilon_1,\varepsilon_2) \,                                                                \\
                                           & =\{\, g(\varepsilon_1)^\circledast\circ(f(\varepsilon_1,\varepsilon_4))\circ (\varepsilon_1,\varepsilon_2)\,\}                                                                         \\
    \partial_{f(f(a,a),a)}                 & = \partial_{f(\varepsilon_1,\varepsilon_2)}(\partial_{f(a,a)}(\partial_a(E)))                                                                                                          \\
                                           & =\{\,g(\varepsilon_1)^\circledast \,\}.
  \end{align*}
\end{example}

Let us now show that the computation of an automaton from the partial derivation cannot be achieved \emph{via} the same algorithm as the case of derivation.
Let us first try to extend the previous algorithm.
Given a tree expression \(E\) over an alphabet \(\Sigma \), we first compute the set
\begin{equation*}
  D_0(E)=\bigcup_{a\in\Sigma_0} \partial_a(E).
\end{equation*}
From this set, we compute the tree automaton \(A_0=(\Sigma,D_0(E),F_0,\delta_0)\) where
\begin{align*}
  F_0      & = \{E'\in D_0(E) \mid \varepsilon_1 \in L(E')\},        &
  \delta_0 & = \{(a, E') \mid a \in \Sigma_0, E'\in\partial_a(E) \}.
\end{align*}
From this step, one can choose a total function \(\mathrm{tree}_0\) associating any tree expression \(E'\) in \(D_0(E)\) with a tree \(t\) such that
\begin{equation*}
  \mathrm{tree}_0(E')= t \Rightarrow  E'\in \partial_t(E).
\end{equation*}
Notice that one can consider a \emph{greedy} choice, that is a choice that minimizes the number of trees in the codomain of the function \(\mathrm{tree}_0\).

From this induction basis, let us consider the transition set \(\delta_n\) inductively defined by
\begin{align}
  \begin{split}
    \delta_n = \{((E'_1, \ldots, E'_m), f, E') \mid
    & E'\in\partial_{f(t_1,\ldots,t_m)}(E), t_i = \mathrm{tree}_{n-1}(E'_i),\\
    & f \in \Sigma_m, E'_1,\ldots,E'_m \in D_{n-1}(E)\}.
    % \label{eq:delta}
  \end{split}
\end{align}
Let us consider the set
\(
D_n(E) = D_{n-1}(E) \cup \pi_3(\delta_n)
\),
where \(\pi_3\) is the classical projection defined by
\(
\pi_3(X) = \{z \mid (\_,\_,z) \in X\}
\).
Obviously, one can once again choose a total function \(\mathrm{tree}_n\) associating any tree expression \(E'\) in \(D_n(E)\) with a tree \(t\) such that
\begin{equation*}
  \mathrm{tree}_n(E') = t\Rightarrow E'\in \partial_t (E),
\end{equation*}
by choosing a transition \(((E'_1, \ldots, E'_m), f, E')\) in \(\delta_n\) for any tree expression \(E'\) in \(D_{n}(E)\setminus (D_{n-1}(E))\)
and defining \(t\) as \(f(\mathrm{tree}_{n-1}(E'_1), \ldots, \mathrm{tree}_{n-1}(E'_k))\).
Finally, by considering the set
\begin{equation}
  F_n = \{E' \in D_n(E) \mid \varepsilon_1 \in L(E')\},
  % \label{eq:final states}
\end{equation}
we can define the tree automaton \(A_n=(\Sigma,D_n(E),F_n,\delta_n)\).

Let us consider the expression \(E = f(a,a) + f(a,b) + f(b,a )\).
Then
\begin{align*}
  \delta_a(E) & = \{ f(\varepsilon_1,a), f(a, \varepsilon_1), f(\varepsilon_1,b), f(b, \varepsilon_1)\}, \\
  \delta_b(E) & = \{ f(\varepsilon_1,a), f(a, \varepsilon_1)\},                                          \\
  D_0         & = \{ f(\varepsilon_1,a), f(a, \varepsilon_1), f(\varepsilon_1,b), f(b, \varepsilon_1)\}.
\end{align*}
Let us choose \(\mathrm{tree}_0(E') = a\) for any expression \(E'\) in \(D_0\).
By construction, the set \(\delta_1\) contains the transition
\(((f(\varepsilon_1,a), f(a, \varepsilon_1)),f,\varepsilon_1)\), since
\(\mathrm{tree}_0(f(\varepsilon_1,a)) = \mathrm{tree}_0(f(a, \varepsilon_1)) = a\) and since \(\partial_{f(a,a)}(E)=\{\varepsilon_1\}\).
In this case, since there also exist the transitions \((b, f(\varepsilon_1,a))\) and \((b, f(a, \varepsilon_1))\), the tree \(f(b, b)\) is recognized by the automaton, that exhibits a witness of the difference with \(L(E)\).

\section{Web Application}\label{sec:appli}
The computation of a derivative and partial derivative, and both the construction of a derivative tree automaton and the classical non deterministic inductive construction have been implemented in Haskell (made in Haskell, compiled
in Javascript using the \href{https://github.com/reflex-frp/reflex-platform}{\textsc{reflex platform}}, represented with \href{https://github.com/mdaines/viz.js}{\textsc{viz.js}}) in order to help the reader to manipulate the notions.
This web application can be found \href{http://ludovicmignot.free.fr/programmes/BottomUpPartialDerivatives/index.html}{here}~\cite{AppWeb}.
As an example, the tree expression \(\neg({g(a)}^{*_a}) \cdot_a f(f(a,a),a)\) of the examples can be defined from the literal input \texttt{\(\neg\)(g[a]*a).af[f[a,a],a]}.

\section{Conclusion and Perspectives}

We have shown how to compute a derivative tree automaton as a fixed point of an inductive construction when it exists.
Even when it does not exist, the process can be used in order to solve the membership test (\emph{i.e.} does a tree belong to the language denoted by a tree expression?): it is easy to see that, for a tree \(t\) of height~\(h\), the tree automaton \(A_h\) is sufficient, since \(t\in L(A_h)\) if and only if \(t\in L(A)\).
As an example, consider the tree automaton \(A_1\) of Example~\ref{ex:cons}.
This tree automaton is sufficient to determine that the trees in \(T(\{f,g,b,c\})\) belong to \(L(E)\).
And even the subautomaton of \(A_h\) restricted to the transitions used in order to compute \(\Delta(t)\) is enough.
Moreover, due to the independence of the computations of derivatives, this process can be performed in a parallel/concurrent way.

Furthermore, we can wonder whether the choices of the \(\mathrm{tree}_*\) functions during the computation of the derivative automaton impact the produced automaton.
We conjecture that all of these choices lead to a unique automaton, and the statement of Proposition~\ref{prop Delta t dt} could be replaced by \( \Delta(t) = \{d_t(E)\} \).

Let us notice that this fixed point computation cannot be extended directly to deal with partial derivatives
but the partial derivation can be used to solve the membership test syntactically. This is not  the only loss of the extension of partial derivation from words to trees. It also seems that partial derivation tends to produce more expressions than
derivation. It could be the consequence of the distributions that occur when the product and the composition are applied to sets of expressions in conjunction with the composition closure  of sets of expressions, that seems to cancel the reduction
power  of this process.

Finally, the ACI rules  of the sum used here  are the same as used in the case of words,  the difficulties that we have found reside on the sufficiency of these rules  to  prove that the set of derivatives is  finite. Seen that we deal with symbols of rank $\geq 1 $ and with operations   like  the composition~($\circ$), the $a$-product~($\cdot_a$), the iterated composition ($\circledast$) which are somehow ``complicated''. However, all the examples that we have made (taking into account these rules) halt. So for now, this is only an hypothesis to justify. The study of the finiteness of the set of (similar) derivatives is the next step of our study: are the three ACI rules sufficient to obtain a finite set of derivatives?
Moreover, the same question arises as far as partial derivatives are concerned: unlike the word case, does the partial derivation  need reduction rules to produce a finite set of derived terms from an expression?
\bibliographystyle{plain}
\bibliography{biblio}

% \textbf{Ajouter le rafinement possible dans le Lemme 1: on peut penser que \( \Delta(t) = \{d_t(E)\} \), ce qui implique le Lemme 1, et qui fait que le choix des fonctions \(\mathrm{tree}_*\) n'a pas d'importance, toutes amènent au même automate.}
%%-----------------------------
%%      your bibliography
%%-----------------------------
\end{document}